\documentclass[12pt, a4paper, fleqn]{article}

\usepackage{amsfonts}
\usepackage{bbm}
\usepackage{amsmath}
\usepackage{amssymb}
\usepackage[english]{babel}
\usepackage{amsthm}
\usepackage{tikz}
\usetikzlibrary{arrows.meta, calc, patterns}

\newtheorem {thm}{Theorem}[section]
\newtheorem {lem}[thm]{Lemma}

\newtheorem {prop}[thm]{Proposition}

\theoremstyle{definition}
\newtheorem {defn}[thm]{Definition}

\theoremstyle{remark}
\newtheorem*{rem*}{Remark}

\DeclareMathOperator{\AnnCross}{AC}

\DeclareMathOperator{\dist}{dist}

\DeclareMathOperator{\cl}{cl}
\DeclareMathOperator{\Dfc}{Dfc}
\DeclareMathOperator{\conv}{conv}

\DeclareMathOperator*{\argmin}{argmin}
\DeclareMathOperator{\vrt}{vert}

\title{On percolation of two-dimensional hard disks}
\author{Alexander Magazinov\thanks{Supported by ERC Starting Grant 639305 and ERC Starting Grant 678520.}}

\begin{document}
\maketitle

\begin{abstract}
We consider the hard-core model in $\mathbb{R}^2$, in which a random set of non-intersecting unit disks is sampled with an intensity parameter $\lambda$. Given $\varepsilon>0$ we consider the graph in which two disks are adjacent if they are at distance $\leq \varepsilon$ from each other. We prove that this graph, $G$, is highly connected when $\lambda$ is greater than a certain threshold depending on $\varepsilon$.
Namely, given a square annulus with inner radius $L_1$ and outer radius $L_2$, the probability that the annulus is crossed by $G$ is at least $1 - C \exp(-cL_1)$. As a corollary we prove that
a Gibbs state admits an infinite component of $G$ if the intensity $\lambda$ is large enough, depending on $\varepsilon$.
\end{abstract}

\section{Introduction}

How closely can one pack spheres in space? This basic question has fascinated many mathematicians over the years, motivated by its simplicity and various practical applications. Kepler conjectured in 1611 that the most efficient packing of cannonballs in 3 dimensions is given by the Face-Centered Cubic (FCC) lattice (though this arrangement is not unique). Lagrange proved in 1773 that the most efficient \emph{lattice} packing in 2 dimensions is given by the triangular lattice and this was extended by Gauss in 1831 to show that the FCC lattice achieves the highest \emph{lattice} packing density.

The proof that the triangular lattice packing is optimal among all packings is generally attributed to Thue~\cite{Thu}. A standard proof of the two-dimensional
analogue of Kepler conjecture is now considered folklore, and can be found in~\cite{FT}.

Fejes T\'oth suggested in 1953 a method for verifying the Kepler conjecture by checking a finite number of cases. This method was finally applied successfully by Hales in 1997 in his groundbreaking, computer-assisted, proof of the Kepler conjecture.

More recently, Cohn and Elkies~\cite{CE} explained how the existence of functions with certain Fourier-analytic properties can be used to give an upper bound on the density of packings. This bound was applied by Viazovska last year in a crowning achievement to prove that the $E_8$ lattice gives the densest packing in 8 dimensions; a work which was immediately extended to prove that the Leech lattice gives the densest packing in 24 dimensions, both long-standing conjectures (see~\cite{Vi1, Vi2}).

The densest packing is, in a sense, perfect, lacking in holes or defects. How do near-optimal packings look like? One natural way to make sense of this question is to consider random packings with high density. A standard model for such random packings is the hard-sphere model described below. The model is parameterized by an intensity parameter $\lambda$ with the average density of the random packing increasing with $\lambda$ in a way that the packings with the maximal density arise formally in the limit $\lambda\to\infty$. It is natural to expect that very dense packings would preserve some of the structure of the densest packings. For instance, that in 2 dimensions they would resemble the triangular lattice and that in 3 dimensions they would resemble the FCC lattice or one of the other packings of maximal density. While this may be true locally, it was proved by Richthammer~\cite{Ri} in a significant breakthrough that configurations sampled from the two-dimensional hard-disk model at high intensity are globally rather different from the triangular lattice in that they lose the translation rigidity of the latter over long distances.

It remains a major challenge to understand in what ways are typical configurations of the high-intensity hard-disk model similar, or dissimilar, to the densest packings. For instance, is rotational rigidity preserved in typical configurations in two dimensions? In statistical physics terminology one is interested in characterizing the translation-invariant Gibbs states of the hard-disk model and showing that multiple states exist, equivalently that symmetry breaking occurs, at high intensities. No proof of symmetry breaking is currently known in any dimension.

Lyons, Bowen, Radin and Winkler put forward the question of percolation in the hard-disk model. Namely, given a packing of unit disks, let us connect by an edge any two centers whenever the distance between
them is at most $2 + \varepsilon$. Does the resulting graph contain an infinite connected component? Bowen, Lyons, Radin and Winkler explicitly conjectured~\cite[Section 7]{BLRW} that the percolation indeed
occurs almost surely in the two-dimensional case for any model obtained as a weak-$*$ limit ($N \to \infty$) of $N\mathbb Z^2$-periodic packings of density $d_N$,
where $\lim\limits_{N \to \infty} d_N = d(\varepsilon) < \frac{1}{2\sqrt{3}}$.

Aristoff~\cite{Ar} addressed a similar question to that of~\cite{BLRW}: does percolation at distance $2 + \varepsilon$ occur almost surely for Gibbs states of the hard-disk model if the intensity parameter
$\lambda$ satisfies $\lambda > \lambda(\varepsilon)$? Aristoff managed to provide a positive answer for $\varepsilon > 1$. Passing from intensity to density is indeed relevant; the former has control over the
latter (see~\cite[Subsection 2.2 and Appendix A]{MMSWD}).

In this paper we focus on generalizing of Aristoff's result by eliminating the restriction $\varepsilon > 1$. This result strongly supports the conjecture by Bowen, Lyons, Radin and Winkler, although
the definition of a model just by density, as in~\cite{BLRW}, may, a priori, be less restrictive than the definition via the Grand canonical ensemble.

We also menton that there are numerous papers considering similar question for another models of statistical mechanics. See, for example~\cite{BL,Jan,Stu}.

\section{Main result}
The main result of this paper is concerned with the Poisson hard-disk model in the Euclidean space
$\mathbb{R}^d$. In order to give the definition of this model, we first recall the notion of
a Poisson point process.

Throughout the paper the following notation is used: $\# X$ denotes the cardinality of a finite set $X$;
$|A|$ denotes the $d$-dimensional Lebesgue measure of a set $A \subseteq \mathbb R^d$ (whenever we use this notation
the dimension is implicit from the context); $\| v \|$ denotes the Euclidean norm of a vector $v \in \mathbb R^d$,
correspondingly, $\| x - y \|$ denotes the Euclidean distance between two points $x, y \in \mathbb R^d$.

\begin{defn}(See~\cite[Section 2.1]{Kin}.)
Let $\lambda > 0$. Assume that $\mu$ is a random measure on the algebra $\mathcal B(\mathbb R^d)$ of
all Borel sets in $\mathbb R^d$. Let the following properties be satisfied:
\begin{enumerate}
  \item For set $A \in \mathcal B(\mathbb R^d)$ with $|A| < \infty$ the random variable
        $\mu(A)$ has the Poisson distribution with parameter $\lambda |A|$.
        Equivalently, $\Pr(\mu(A) = i) = \exp(-\lambda|A|) \frac{\lambda^i}{i!}$ for every
        non-negative integer $i$.
  \item For every two sets $A, B \in \mathcal B(\mathbb R^d)$ of finite Lebesgue measure,
        that are disjoint (i.e, $A \cap B = \varnothing$), the random variables $\mu(A)$ and $\mu(B)$ are independent.
\end{enumerate}
Then the measure $\mu$ is called a {\it Poisson point process} with intensity $\lambda$.
\end{defn}

A random measure $\mu$ as above can be uniquely identified with a random discrete point set
$\eta \subseteq \mathbb R^d$ so that
$$\mu(A) = \# (\eta \cap A)$$
for every set $A \in \mathcal B(\mathbb R^d)$.
Hence, from now on, the term {\it Poisson point process} will refer to a random point set.
If $D \subset \mathbb R^d$ is a bounded open set, then the {\it Poisson point process of intensity
$\lambda$ in $D$} is a random set $\eta \cap D$, where $\eta$ is a Poisson point process of of intensity
$\lambda$ in $\mathbb R^d$.

The point arrangements appearing in the hard-disk model are picked from a specific space of point arrangements on $\mathbb{R}^d$.
This space is introduced in the following Definition~\ref{def:config}.

\begin{defn}\label{def:config}
Let
$$ \Omega(\mathbb R^d) = \{ \xi \subset \mathbb R^d: \text{$\| x - y \| > 2$ for all $x, y \in \xi$, $x \neq y$} \}. $$
Then every element $\xi \in \Omega(\mathbb R^d)$ is called a {\it configuration}.
\end{defn}

By Definition~\ref{def:config}, each configuration $\xi$ can be identified with a packing of unit balls
$\{ B_1(x) : x \in \xi \}$. Here and further $B_{\rho}(x)$ denotes a ball of radius $\rho$ centered at $x$.
The condition $\| x - y \| > 2$ is often referred to as {\it interaction between points via hard-core exclusion}.

We are ready to present the definition of the Poisson hard-disk model.

\begin{defn}[Poisson hard-disk model]
Let $D\subseteq\mathbb{R}^d$ be a bounded open set and let $\zeta \in \Omega(\mathbb R^d)$.
Assume that $\lambda$ be a positive real number. Consider the Poisson point process
$\eta$ in $D$ with intensity $\lambda$. Let $\bar{\eta}$ be the conditional distribution of $\eta$ restricted to the event
$$\eta \cup (\zeta \setminus D) \in \Omega(\mathbb R^d).$$
Then the random point set $\eta^{[\lambda]}(D, \zeta)$ defined by
$$\eta^{[\lambda]}(D, \zeta) = \bar{\eta} \cup (\zeta \setminus D)$$
is called the {\it Poisson hard-disk model} on $D$ with intensity $\lambda$ and boundary
conditions $\zeta$.
\end{defn}

Indeed, $\bar{\eta}$ is well-defined, since
$$\Pr (\eta \cup (\xi \setminus D) \in \Omega(\mathbb R^d)) \geq \Pr(\eta = \varnothing) > 0.$$
Also, by definition, $\eta^{[\lambda]}(D, \xi) \in \Omega(\mathbb R^d)$ a.s.

Our main focus is on the existence of large connected components in
the hard-disk model, when we connect two points if their distance is
at most $2+\varepsilon$. For the associated ball packings it means that
the centers of two balls are connected whenever the distance between the balls
does not exceed $\varepsilon$.

We formalize this concept in the next definitions.

\begin{defn}
Let $\xi \in \Omega(\mathbb R^d)$ and $\varepsilon > 0$. Consider the graph $G_{\varepsilon}(\xi) = (V_\varepsilon(\xi), E_\varepsilon(\xi))$,
where
$$
  V_\varepsilon(\xi) = \xi \quad\text{and}\quad
  E_\varepsilon(\xi)=\{\{x,y\} : \text{$x,y \in \xi$, $x\neq y$ $\|x-y\| \leq  2+\varepsilon$} \}.
$$
We will call $G_{\varepsilon}(\xi)$ the {\it connectivity graph} of $\xi$ at distance $\varepsilon$.
\end{defn}

\begin{defn}
Let $\varepsilon > 0$ and $0 < L_1 < L_2$ be fixed. We say that a configuration $\xi \in \Omega(\mathbb R^2)$
{\it is annulus-crossing} if $\xi \in \AnnCross(\varepsilon, L_1, L_2)$, where
\begin{equation*}
  \begin{split}
    \AnnCross(\varepsilon, L_1, L_2) =  \{ & \zeta \in \Omega(\mathbb R^d) : \text{there exist $x, y \in V_{\varepsilon}(\zeta)$} \\
    & \text{such that $x\in(-L_1, L_1)^d$, $y\notin (-L_2, L_2)^d$ and} \\
    & \text{some connected component of $G_\varepsilon(\zeta)$ contains} \\
    & \text{both $x$ and $y$} \}.
  \end{split}
\end{equation*}
\end{defn}

For the sake of shortening the notation, we write $Q_L = (-L, L)^d$. The dimension is omitted, since it
will always be clear from the context.

We can now state the main result of the paper.

\begin{thm}[Main Theorem]
For any $\varepsilon > 0$ there exist positive numbers $\lambda_0$,
$c$, $C$ and $L_0$, depending only on $\varepsilon$, such that
the following holds. If $\lambda > \lambda_0$, $0 < L_1 < L_2$ and $\eta = \eta^{[\lambda]}(Q_{L_2 + L_0}, \zeta)$
is a two-dimensional Poisson hard-disk model, then one necessarily has
$$\Pr\bigl( \eta  \in \AnnCross(\varepsilon, L_1, L_2) \bigr) \geq 1 - C \exp(-cL_1).$$
\end{thm}

\begin{rem*}
We emphasize that the Main Theorem is restricted to the two-dimensional case. One of the steps of our proof essentially relies on
the Jordan theorem implying the intersection of two one-dimensional curves (the so-called ``large circuits'').
\end{rem*}

\begin{rem*}
The statement of the Main theorem involves a ``buffering region'' $Q_{L_2 + L_0} \setminus Q_{L_2}$ which we do not require to be crossed.
Moreover, the width of this region, $L_0$, depends on $\varepsilon$. This generality allows us to avoid many technical difficulties, although
it is likely that, say, $L_0 \equiv 5$ is still sufficient.
\end{rem*}

The Main Theorem deals with configurations with fixed shape outside a finite region (and coincide with the respective boundary conditions).
We will also address a model, the so-called Gibbs distribution, where two samples, say $\xi$ and $\xi'$, typically have unbounded
symmetric difference $\xi \bigtriangleup \xi'$. For results concerning the Gibbs distributions, see Section~\ref{sec:gibbs}.

\section{The defect of a configuration}

Configurations of the hard-sphere model at high intensity $\lambda$ are locally tightly packed, approximating an optimal packing. It will be helpful in the sequel to quantify the local deviation
from an optimal packing. Thus we introduce the notion of a {\it defect}.

\subsection{Definition of the two-dimensional defect function}

We call a configuration $\xi \in \Omega(\mathbb R^d)$ {\it saturated in the
$\rho$-neighborhood} of a bounded open domain $D \subset \mathbb R^d$  if
$$\sup\limits_{y : \dist(y, D) \leq \rho} \dist(y, \xi) \leq 2.$$
If, moreover, the inequality
$$\sup\limits_{y \in \mathbb R^d} \dist(y, \xi) \leq 2$$
holds, then $\xi$ is called {\it saturated (in the entire $\mathbb R^d$)}.
In other words, $\xi$ is saturated if it is a maximal element
of $\Omega(\mathbb R^d)$ with respect to the inclusion.

Let $\xi \in \Omega(\mathbb R^d)$. For every point $x \in \xi$ define a
set $\mathcal V_{\xi}(x) \subseteq \mathbb R^d$ as
follows:
$$\mathcal V_{\xi}(x) = \{ y \in \mathbb R^d : \| y - x \| = \inf\limits_{x' \in \xi} \| y - x' \| \}.$$
$\mathcal V_{\xi}(x)$ is called the {\it Voronoi cell} of $x$ with respect to $\xi$.
The tessellation of $\mathbb R^d$ into Voronoi cells for a given point set $\xi$ is called
the {\it Voronoi tessellation} for $\xi$.

The Voronoi cell $\mathcal V_{\xi}(x)$ is a convex
$d$-dimensional polyhedron, possibly unbounded. Each of its facets
is contained in the perpendicular bisector to a segment $[x, x']$ for some $x' \in \xi$. The next proposition gives a useful way to bound the cells of the Voronoi tesselation.

\begin{prop}\label{prop:voronoi_cell_rad}
Let $\rho > 2$ and $\xi \in \Omega(\mathbb R^d)$. Suppose $\xi$ is saturated in the $\rho$-neighborhood of
a bounded open domain $D$. Then each point $x \in \xi \cap D$ satisfies $\mathcal V_{\xi}(x) \subset B_2(x)$.
\end{prop}

\begin{proof}
Assume the converse: there is
a point $z \in V_{\zeta}(x) \setminus B_2(x)$. Without loss of generality,
one can additionally assume that $\| z - x \| \leq \rho$. Then
$$\sup\limits_{y : \dist(y, D) \leq \rho} \dist(y, \xi) \geq \| z - x \| > 2.$$
This contradicts the assumption that $\xi$ is saturated in the $\rho$-neighborhood of $D$.
\end{proof}

We are ready to proceed with the two-dimensional case. The optimal packing of unit disks in the plane is unique with centers of the disks arranged
as a regular triangular lattice (the edge of a generating triangle equals 2). The Voronoi tesselation of this lattice is the tiling of the plane into equal regular hexagons of area $2\sqrt{3}$ (see~\cite{FT}). Moreover, in a Voronoi tesselation for any configuration, the volume of each cell must exceed $2\sqrt{3}$ (see Proposition~\ref{prop:key_defect_prop} below). This motivates the following definition.

\begin{defn}(Defect of a two-dimensional configuration)
Let $D \subset \mathbb R^2$ be a bounded open domain and let
$\xi, \xi' \in \Omega(\mathbb R^2)$ be two configurations. Suppose that
$\xi \subseteq \xi'$ and that $\xi'$ is locally saturated in the
$\rho$-neighborhood of $D$ for some $\rho > 2$. The {\it defect} $\Delta(\xi, \xi', D)$ of $\xi$
in the domain $D$ with respect to the locally saturated extension $\xi'$ is defined by
\begin{equation}\label{eq:defect_defn}
\Delta(\xi, \xi', D) = \sum\limits_{x \in \xi' \cap D} \left( |\mathcal V_{\xi'}(x)| - 2\sqrt{3} \cdot \mathbbm 1_{\xi}(x) \right).
\end{equation}
\end{defn}

\begin{rem*}
There is no equality case for the inequality
$|\mathcal V_{\xi'}(x)| \geq 2\sqrt{3}$, since the optimal packing of balls does not correspond to an element of $\Omega(\mathbb R^2)$: the strict inequality $\| x - y \| > 2$ is violated.
\end{rem*}

\subsection{Properties of the defect function}

In this section we give an abstract definition of a defect function by listing the properties that it should satisfy. This definition applies equally well in any dimension.
It will be checked in the next section that the two-dimensional defect function introduced above is consistent with the abstract definition.

\begin{rem*}
We choose to give an abstract definition, since there is a chance to generalize the Main Theorem to higher dimensions within the same framework. Indeed, the key intermediate lemma, the Thin Box Lemma,
relies exclusively on the properties included in Definition~\ref{def:defect_prop}. On the other hand, there is a construction in the 3-dimensional space satisfying these properties. The construction
is based on Hales' analysis of the optimal packing there (see~\cite[Theorems 1.5, 1.7, 1.9]{Ha}). However, the details are left beyond the scope of this paper as we focus on the two-dimensional case.
The author is not aware of any constructions of defect functions in dimensions $d \geq 4$.
\end{rem*}

\begin{defn}
The {\it optimal packing density} $\alpha(d)$ in $d$ dimensions is defined by
$$\alpha(d) = \limsup\limits_{L \to \infty} \sup\limits_{\xi \in \Omega(\mathbb R^d)} \frac{\# (\xi \cap Q_L)}{|Q_L|}.$$
\end{defn}

\begin{defn}
A {\it defect-measuring triple} is a tuple $(\xi, \xi', D)$ where $\xi' \in \Omega(\mathbb R^d)$, $\xi \subseteq \xi'$ and $D$ is a bounded open domain in $\mathbb R^d$.
\end{defn}

\begin{defn}\label{def:defect_prop}
Let $\varepsilon>0$. Let $F$ be a partial function taking defect-measuring triples as arguments and returning real numbers.
We say that $F$ {\it satisfies the defect function properties at level $\varepsilon$}
if there exist real numbers $c, C_{cnt} > 0$ and $\rho \geq 100$ such that the following holds.
\begin{enumerate}
    \item (Domain of definition.) If $(\xi, \xi', D)$ is a defect-measuring triple and $\xi'$ is saturated
          in the $\rho$-neighborhood of $D$, then $F(\xi, \xi', D)$ is necessarily defined.
    \item (Localization.) If $(\xi_1, \xi'_1, D)$ and $(\xi_2, \xi'_2, D)$ are two defect-measuring triples such that
          \begin{align*}
            \xi_1 \bigtriangleup \xi_2 = \varnothing \quad & \text{or} \quad \dist(\xi_1 \bigtriangleup \xi_2, D) > \rho \quad \text{and} \\
            \xi'_1 \bigtriangleup \xi'_2 = \varnothing \quad & \text{or} \quad \dist(\xi'_1 \bigtriangleup \xi'_2, D) > \rho,
          \end{align*}
	      then either $F$ is defined on both triples and
          $$F(\xi_1, \xi'_1, D) = F(\xi_2, \xi'_2, D),$$
          or $F$ is undefined on both triples.
    \item (Positivity.) If $F(\xi, \xi', D)$ is defined, then $F(\xi, \xi', D) \geq 0$.
	\item (Monotonicity.) If $F(\xi, \xi', D)$ is defined and $D' \subseteq D$ is an open domain, then
	      $$F(\xi, \xi', D') \leq F(\xi, \xi', D).$$
	\item (Additivity.) If $D_1, D_2 \subset \mathbb R^d$ are two bounded open domains, $D_1 \cap D_2 = \varnothing$ and
		    $F(\xi, \xi', D_1 \cup D_2)$ is defined, then both $F(\xi, \xi', D_1)$ and $F(\xi, \xi', D_2)$ are defined, and
		    $$F(\xi, \xi', D_1 \cup D_2) = F(\xi, \xi', D_1) + F(\xi, \xi', D_2).$$
	\item (Saturation.) If $(\xi, \xi', D)$ is a defect-measuring triple and the inequality $F(\xi, \xi', D) < c$ holds,
          then one necessarily has
	      $$\xi \cap D = \xi' \cap D.$$
	\item (Connectivity.) If $D$ is convex and a defect-measuring triple $(\xi, \xi', D)$ satisfies $F(\xi, \xi', D) < c$,
	      then there is a connected component $\mathfrak c \subseteq G_{\varepsilon}(\xi \cap D)$ such that
          $$\{y \in \xi : B_{\rho}(y) \subseteq D \} \subseteq \vrt(\mathfrak c).$$
	\item (Distance-decreasing step.) If a defect-measuring triple $(\xi, \xi', D)$, a point $x \in \xi$ and a point $y \in \mathbb R^d$
          satisfy the conditions
          \begin{align*}
            & F(\xi, \xi', D) < c, \\
            & B_{2 + \varepsilon / 2}(x) \subseteq D, \\
            & \| x - y \| > \rho,
          \end{align*}
          then there exists a point $x' \in \xi$ such that
		  $$\| x - x' \| < 2 + \frac{\varepsilon}{2} \quad \text{and} \quad \| x' - y \| < \| x - y \|.$$
	\item (Forbidden distances.) If the inequality $F(\xi, \xi', D) < c$ holds for a defect-measuring triple $(\xi, \xi', D)$,
	      then every two points $x, y \in \xi \cap D$ satisfy
		  $$\| x - y \| \notin [2 + 0.9\varepsilon, 2 + \varepsilon].$$
	\item (Point counting.) If $(\xi, \xi', Q_L)$ is a defect-measuring triple on which $F$ is defined then
	      $$F(\xi, \xi', Q_L) \leq |Q_L| - \frac{1}{\alpha(d)} \cdot \# (\xi \cap Q_L) + C_{cnt} L^{d - 1}.$$
\end{enumerate}
\end{defn}

We proceed by formulating the key result of this section.

\begin{lem}[Defect Lemma]
There exists $\varepsilon_0 > 0$ such that the two-dimensional defect function $\Delta(\xi, \xi', D)$
defined by~\eqref{eq:defect_defn} satisfies the defect function properties at every level $\varepsilon \in (0, \varepsilon_0)$.
\end{lem}

\subsection{Proof of the Defect Lemma}

The proof is based on the following Proposition~\ref{prop:key_defect_prop}.

\begin{prop}\label{prop:key_defect_prop}
The following assertions are true.
\begin{enumerate}
  \item Let $\xi \in \Omega(\mathbb R^2)$ and $x \in \xi$. Assume that the Voronoi cell $\mathcal V_{\xi}(x)$ is bounded. Then $|\mathcal V_{\xi}(x)| \geq 2 \sqrt{3}$.
  \item There exists $c > 0$ such that the following holds. If a configuration $\xi \in \Omega(\mathbb R^2)$ and a point $x \in \xi$ satisfy $|\mathcal V_{\xi}(x)| < 2 \sqrt{3} + c$,
        then $\mathcal V_{\xi}(x)$ is necessarily a hexagon.
  \item For every $\delta > 0$ there exists $c(\delta) > 0$ such that the following holds. If a configuration $\xi \in \Omega(\mathbb R^2)$ and a point $x \in \xi$ satisfy
        $|\mathcal V_{\xi}(x)| < 2 \sqrt{3} + c(\delta)$, then there is a regular hexagon $H$ of area $2 \sqrt{3}$ centered at $x$ such that $d_{Haus}(\mathcal V_{\xi}(x), H) < \delta$.
        (The notation $d_{Haus}(\cdot, \cdot)$ denotes the Hausdorff distance between planar convex bodies.)
\end{enumerate}
\end{prop}

\begin{proof}
See~\cite[Chapter 3]{FT} or~\cite{Rog}.
\end{proof}

We are ready to prove the Defect Lemma.

\begin{proof}[Proof of the Defect Lemma] Let us address each defect function property.

\noindent {\bf 1. Domain of definition.} We show that any choice of $\rho \geq 100$ is sufficient to fulfill this condition. Indeed, for every $x \in \xi' \cap D$ the Voronoi cell
$\mathcal V_{\xi'}(x)$ is bounded by Proposition~\ref{prop:voronoi_cell_rad}, consequently, the corresponding summand $|\mathcal V_{\xi'}(x)| - 2\sqrt{3} \cdot \mathbbm 1_{\xi}(x)$ is defined. Further,
the set $\xi' \cap D$ is finite because $\xi' \in \Omega(\mathbb R^2)$ and $D$ is bounded. Hence the sum in the right-hand side of~\eqref{eq:defect_defn} is finite and therefore well-defined.

\noindent {\bf 2. Localization.} Since $\xi'_1 \cap D = \xi'_2 \cap D$, we have
$$\Delta(\xi_i, \xi'_i, D) = \sum\limits_{x \in \xi'_1 \cap D} f_i(x), \quad (i = 1, 2),$$
where
$$f_i(x) = |\mathcal V_{\xi'_i}(x)| - 2\sqrt{3} \cdot \mathbbm 1_{\xi_i}(x).$$
We claim that $f_1(x) = f_2(x)$ for each $x \in \xi'_1 \cap D$. Indeed,
$\mathbbm 1_{\xi_1}(x) = \mathbbm 1_{\xi_2}(x)$ because $\xi_1 \cap D = \xi_2 \cap D$.

Next, let us show that $\mathcal V_{\xi'_1 \cap B_4(x)}(x) = \mathcal V_{\xi'_2 \cap B_4(x)}(x)$. Indeed, assume the converse.
Then $(\xi'_1 \bigtriangleup \xi'_2) \cap B_4(x) \neq \varnothing$. This contradicts the assumption $\dist(\xi'_1 \bigtriangleup \xi'_2, D) > \rho$.

Therefore
$$\mathcal V_{\xi'_1}(x) = \mathcal V_{\xi'_1 \cap B_4(x)}(x) = \mathcal V_{\xi'_2 \cap B_4(x)}(x) = \mathcal V_{\xi'_2}(x),$$
where the first and the third identities hold because $\rho > 2$ and the configurations $\xi'_i$ are saturated in the $\rho$-neighborhood of $D$.

Hence the expressions for $\Delta(\xi_1, \xi'_1, D)$ and $\Delta(\xi_2, \xi'_2, D)$ are tautologically identical.

\noindent {\bf 3. Positivity.} Follows immediately from Proposition~\ref{prop:key_defect_prop}, assertion 1.

\noindent {\bf 4. Monotonicity.} Follows immediately from Proposition~\ref{prop:key_defect_prop}, assertion 1.

\noindent {\bf 5. Additivity.} Follows immediately from the definition of $\Delta(\xi, \xi', D)$.

\noindent {\bf 6. Saturation.} We show that any choice of $c < 2 \sqrt{3}$ is sufficient. Indeed, if there exists $x \in (\xi' \cap D) \setminus (\xi \cap D)$, then
$$\Delta(\xi, \xi', D) \geq |\mathcal V_{\xi'}(x)| - 2\sqrt{3} \cdot \mathbbm 1_{\xi}(x) = |\mathcal V_{\xi'}(x)| \geq 2 \sqrt{3}.$$
(Here both inequalities follow from Proposition~\ref{prop:key_defect_prop}, assertion 1.)

\noindent {\bf 7. Connectivity.} We claim that, with any choice of $\rho \geq 100$, there exists $c_0 \in (0, 2 \sqrt{3})$ such that any choice $c \in (0, c_0)$ fulfills the Connectivity
property.

Denote $D' = \{ y \in \mathbb R^2 : B_{\rho}(y) \subseteq D \}$. Clearly,
$D'$ is an open bounded convex set.

Let $x, x' \in \xi \cap D'$.
Consider the segment $[x, x']$. With a small perturbation of this segment
we can obtain a curve segment $\gamma \subset D'$, connecting $x$ and $x'$, such that $\gamma$ avoids all vertices of the Voronoi tessellation
for $\xi'$.

Consider the sequence of points $x_0 = x, x_1, \ldots, x_k = x'$ ($x_i \in \xi'$) such that $\gamma$ intersects $\mathcal V_{\xi'}(x_0)$, $\mathcal V_{\xi'}(x_1)$, $\ldots$, $\mathcal V_{\xi'}(x_k)$
in that order. (Consequently, the cells $\mathcal V_{\xi'}(x_i)$ and
$\mathcal V_{\xi'}(x_{i + 1})$ share a common edge.)

Applying Proposition~\ref{prop:voronoi_cell_rad} yields $\rho > 2 \geq \dist(x_i, \gamma) \geq \dist(x_i, D')$. Therefore $x_i \in \xi' \cap D$. Moreover, $x_i \in \xi \cap D$, as implied by an assumption $c < c_0 < 2\sqrt{3}$ and the already proved Saturation property.

Let us prove that an appropriate choice of $c_0$ implies $\| x_i - x_{i + 1} \| < 2 + \varepsilon$. Indeed, we have $| \mathcal V_{\xi'}(x_j) | < 2 \sqrt{3} + c_0$ ($j = i, i + 1$), hence
by choosing $c_0$ sufficiently small we can ensure that both $\mathcal V_{\xi'}(x_i)$ and $\mathcal V_{\xi'}(x_{i + 1})$ are hexagons. (Here we use Proposition~\ref{prop:key_defect_prop}, assertion 2.) Consider the common
edge $[v, w]$ of the cells $\mathcal V_{\xi'}(x_i)$ and $\mathcal V_{\xi'}(x_{i + 1})$. Then
\begin{equation}\label{eq:triangle}
\| x_i - x_{i + 1} \| \leq \left\| x_i - \frac{v + w}{2} \right\| +
\left\| x_{i + 1} - \frac{v + w}{2} \right\|.
\end{equation}
By Proposition~\ref{prop:key_defect_prop}, assertion 3, the summands
on the right-hand side of~\eqref{eq:triangle} are sufficiently close to 1
if $c_0$ is small enough. Therefore $(x_i, x_{i + 1})$ is an edge of
$G_{\varepsilon}(\xi \cap D)$.

Consequently, $x x_1 x_2 \ldots x_{k - 1} x'$ is a path in $G_{\varepsilon}(\xi \cap D)$. Hence the Connectivity property follows.

\noindent {\bf 8. Distance-Decreasing Step.} Let us prove this property under the assumption $\varepsilon < \varepsilon_0 \leq 0.1$.

We claim that, with any choice of $\rho > 100$, there exists $c_0 \in (0, 2 \sqrt{3})$ such that any choice $c \in (0, c_0)$ fulfills the Distance-Decreasing Step
property. Similarly to the previous argument, we can guarantee that $\mathcal V_{\xi'}(x)$ is a hexagon. Denote by $x_1, x_2, \ldots, x_6$ those points in $\xi'$ with the respective Voronoi cells $\mathcal V_{\xi'}(x)$
sharing common edges with $\mathcal V_{\xi'}(x)$. Then
\begin{equation}\label{eq:neighbors}
\mathcal V_{\xi'}(x) = \mathcal V_{\{x, x_1, x_2, \ldots, x_6 \}}(x).
\end{equation}

By Proposition~\ref{prop:key_defect_prop}, assertion 3, an appropriate choice of $c_0$ guarantees
$$d_{Haus}(\conv \{x_1, \ldots, x_6 \}, H) < \frac{\varepsilon}{10},$$
where $H$ is some regular hexagon with circumradius 2 centered at $x$. Consequently, $2 < \| x - x_i \| < 2 + \frac{\varepsilon}{2}$, which, in turn, implies $x_i \in D$ and $x_i \in \xi$.

Since $\varepsilon < 0.1$, the rays from $x$ to all $x_i$ split the plane into six angles with each of the angles not exceeding $70^\circ$. Thus, with no loss of generality, we can assume that
the angle $\angle yxx_1$ (i.e., the angle between the vectors $y - x$ and $x_1 - x$) does not exceed $35^\circ$. Therefore
$$\|y - x_1 \|^2 \leq \| y - x \|^2 + \|x - x'\|^2 - 2\| y - x \| \| x - x' \| \cos 35^\circ < \| y - x \|^2,$$
because, indeed,  $2 \| x - x' \| \cos 35^\circ < \| y - x \|$.

Hence $\|y - x_1 \| < \| y - x \|$ and Distance-Decreasing Step property is proved.

\noindent {\bf 9. Forbidden Distances.} Again, we prove this property under the assumption $\varepsilon < \varepsilon_0 \leq 0.1$.

Similarly to the previous argument, an appropriate choice of $c_0$ guarantees the existence of 6 points $x_1, \ldots, x_6 \in \xi'$ such that~\eqref{eq:neighbors} is satisfied.
Then one can check the inclusion
$$B_{2.1}(x) \subseteq B_2(x) \cup B_2(x_1) \cup \ldots \cup B_2(x_6).$$
Therefore every point $x' \in \xi' \setminus \{x, x_1, \ldots, x_6 \}$ satisfies the inequality $\| x' - x \| > 2.1 > 2 + \varepsilon$. On the other hand, $x' \in \{ x, x_1, \ldots, x_6 \}$
implies $\| x' - x \| < 2 + \frac{\varepsilon}{2} < 2 + 0.9\varepsilon$. Hence the Forbidden Distances property follows.

\noindent {\bf 10. Point counting.} We have
$$\bigcup \limits_{x \in (\xi' \cap D)} \mathcal V_{\xi'}(x) \supseteq Q_{L - 2}.$$
Indeed, otherwise there is a point $y$ that belongs to the right-hand side and does not belong to the left-hand side. Then $\xi' \cup \{y\} \in \Omega(\mathbb R^2)$,
which is impossible because $\xi'$ is saturated in the $\rho$-neighborhood of $Q_L$.

Therefore
\begin{multline*}
F(\xi, \xi', Q_L) = \sum\limits_{x \in \xi' \cap D} |\mathcal V_{\xi'}(x)| - 2\sqrt{3} \# ( \xi \cap Q_L ) \geq \\
|Q_{L - 2}| - \frac{1}{\alpha(2)} \# ( \xi \cap Q_L ) = \geq |Q_L| - \frac{1}{\alpha(2)} \# ( \xi \cap Q_L ) - 16L.
\end{multline*}
Hence the Counting property follows.

We now conclude the proof of the lemma. By the argument above, it is sufficient to set $\varepsilon_0 = 0.1$, $\rho(\varepsilon) \equiv 100$, $C_{cnt}(\varepsilon) \equiv 16$. Finally, in order to define
$c(\varepsilon)$ it is sufficient to fulfill the restrictions of the Connectivity, Distance-Decreasing Step and Forbidden Distances properties.
\end{proof}

\section{The Thin Box Lemma}

\subsection{The uniform model}

A Poisson hard-disk model can be represented as a mixture of the so-called {\it uniform hard-disk models},
defined below. In~\cite[Section 2]{Ar} this representation is used as an equivalent definition of a Poisson model.
Therefore most of our auxiliary results will concern the uniform models and then passed to the Poisson model by means of
Lemma~\ref{lem:pois_to_unif} in Section~\ref{sec:main_proof}.

Let us introduce some notation. Given {\it boundary conditions} $\zeta \in \Omega(\mathbb R^d)$ and a bounded
open domain $D \subseteq \mathbb R^d$, we write
$$\Omega(D, \zeta) = \{ \xi \in \Omega(\mathbb R^d) : \xi \setminus D = \zeta \setminus D \}.$$
One can thus notice that $\Omega(D, \zeta)$ is exactly the domain of values for the Poisson hard-disk model in $D$
with boundary conditions $\zeta$ and an arbitrary intensity. Next, given additionally an integer $s \geq 0$, denote
$$\Omega^{(s)}(D, \zeta) = \{ \xi \in \Omega(D, \zeta) : \# (\xi \cap D) = s \}.$$

\begin{defn}[Uniform hard-disk model]
Let $D \subseteq \mathbb R^d$ be a bounded open domain and $\zeta \in \Omega(\mathbb R^d)$ be a
configuration. Assume an integer $s \geq 0$ satisfies $\Omega^{(s)}(D, \zeta) \neq \varnothing$.
Consider the random set $\eta$ of $s$ points sampled independently from the uniform
distribution on $D$. Denote by $\bar{\eta}$ the conditional distribution of $\eta$ restricted to the event
$\eta \cup (\zeta \setminus D) \in \Omega^{(s)}(D, \zeta)$.
Then the random point set $\eta^{(s)}(D, \zeta)$ defined by
$$\eta^{(s)}(D, \zeta) = \bar{\eta} \cup (\zeta \setminus D)$$
is called the $s$-point {\it uniform hard-disk model} on $D$ with boundary conditions $\zeta$.
\end{defn}

One can see that $\bar{\eta}$ is well-defined. Indeed, the inequality
$$\Pr (\eta \cup (\xi \setminus D) \in \Omega^{(s)}(D, \zeta)) > 0$$
follows from the assumption $\Omega^{(s)}(D, \zeta) \neq \varnothing$.

\subsection{Statement of the Thin Box Lemma}

In the previous subsection we defined the uniform hard-disk model. Before that, we gave a definition of an (abstract $d$-dimensional)
defect function. Let us bring this notions together in the next two definitions.
For the rest of this section we assume that the dimension $d \geq 2$, a constant $\varepsilon > 0$ are given, the function $\Delta$
satisfies the defect function properties at level $\varepsilon$. We keep the notation $\rho, c, C_{cnt}$ for the respective constants from Definition~\ref{def:defect_prop}.

\begin{defn}
Let $D_1, D_2 \subseteq \mathbb R^d$ be two bounded open domains, and $D_1 \subseteq D_2$.
Assume a real number $\rho > 100$, boundary conditions $\zeta \in \Omega(\mathbb R^d)$ and an integer
$s \geq 0$ are given. A measurable map $\phi : \Omega^{(s)}(D, \zeta) \to \Omega(\mathbb R^d)$
is called a {\it $(D_2, \rho)$-saturator} over $\Omega^{(s)}(D_1, \zeta)$ if
the following holds for every $\xi \in \Omega^{(s)}(D, \zeta)$:
\begin{enumerate}
  \item $\phi(\xi) \supseteq \xi$.
  \item $\phi(\xi)$ is saturated in the $\rho$-neighborhood of $D_2$.
\end{enumerate}
\end{defn}

\begin{defn}\label{def:bd_defect}
Let two bounded open domains $D_1 \subseteq D_2 \subset \mathbb R^d$, the boundary conditions $\zeta \in \Omega(\mathbb R^2)$
and an integer $s \geq 0$ be given. We say that the defect of the model $\eta^{(s)}(D_1, \zeta)$ with respect to the outer domain $D_2$ {\it is bounded by a
constant} $\Delta_0 > 0$ if there exists a $(D_2, \rho)$-saturator $\phi$ over $\Omega^{(s)}(D_1, \zeta)$ such
that the inequality
$$\Delta(\xi, \phi(\xi), D_2) \leq \Delta_0$$
holds for every $\xi \in \Omega^{(s)}(D_1, \zeta)$.
\end{defn}

In the notation of Definition~\ref{def:bd_defect} let
\begin{multline*}
  \Dfc_{\Delta_0}(D_1, D_2) = \{(\zeta, s) : \text{$\zeta \in \Omega(\mathbb R^2)$, $s \in \mathbb Z$, $s \geq 0$} \\
  \text{and the defect of $\eta^{(s)}(D_1, \zeta)$ with respect to $D_2$ is bounded by $\Delta_0$} \}.
\end{multline*}

Now we turn to specific domains in $\mathbb R^d$. For $K > 100 \rho$, $n \in \mathbb N$, $i \in \mathbb Z$ we will write
$$R(K, n) = (-K, K)^{d - 1} \times (-20nK, 20nK),$$
$$R'(K, n) = (-5K, 5K)^{d - 1} \times (-20nK, 20nK),$$
$$P_i(K) = (-5K, 5K)^{d - 1} \times ((10i - 5)K, (10i + 5)K).$$
We will consider two important classes of configurations defined below.

\begin{defn}\label{def:cross}
We say that a configuration $\xi$ {\it $(\varepsilon, \nu)$-crosses} the box
$R'(K, n)$ if there exists a set of indices $I \subseteq \{-2n + 1, -2n + 2, \ldots, 2n - 1 \}$ such that
\begin{enumerate}
  \item $\# I \geq 4n - 1 - \nu$.
  \item If $i \in I$ and $\xi_i = \{ x \in \xi : B_{\rho}(x) \subseteq P_i(K) \}$
        then $\xi_i \neq \varnothing$ and every two points of $\xi_i$ are
        connected by a path in $G_{\varepsilon}(\xi \cap P_i(K))$.
  \item If $i_1, i_2 \in I$ and two points $x_1, x_2 \in \xi$ satisfy $B_{\rho}(x_k) \subseteq P_{i_k}(K)$ ($k = 1, 2$) then
        $x_1$ and $x_2$ are connected by a path in $G_{\varepsilon}(\xi \cap R'(K, n))$.
\end{enumerate}
\end{defn}

\begin{defn}\label{def:empty}
Let $D \subseteq \mathbb R^d$ be a bounded open domain. We say that a configuration $\xi$ {\it admits an empty $\varepsilon$-space in $D$} if
there exists a point $w \in \mathbb R^d$ such that
$$B_{\varepsilon}(w) \subset D \quad \text{and} \quad \dist(w, \xi) > 2 + \varepsilon.$$
\end{defn}

Introduce the following notation:
$$p_{cross}[\varepsilon, \nu](K, n, \zeta, s) = \Pr(\text{the box $R'(K, n)$ is $(\varepsilon, \nu)$-crossed by $\eta^{(s)}(R(K, n), \zeta)$} ),$$
$$p_{empty}[\varepsilon](K, n, \zeta, s) = \Pr(\text{$\eta^{(s)}(R(K, n), \zeta)$ admits an empty $\varepsilon$-space in $R(K, n)$}).$$

The main result of this section is as follows.

\begin{lem}[Thin Box Lemma]
Let $d \geq 2$ be a fixed dimension. Assume that a function $\Delta$ satisfies the defect function properties at
some fixed level $\varepsilon \in (0, 1)$, and $\rho, c$ are the respective constants from Definition~\ref{def:defect_prop}.
Then there exists $K > 100 \rho$ such that the inequality
\begin{multline*}
\inf\limits_{\substack{(n, \zeta, s) \\ \text{\rm appropriate}}} \left( p_{cross}\left[ \varepsilon, \frac{6 \Delta_0}{c} \right](K, n, \zeta, s) + p_{empty}[\varepsilon](K, n, \zeta, s) \right) > 0, \\
\text{where} \qquad \text{$(n, \zeta, s)$ is appropriate} \Longleftrightarrow (\zeta, s) \in \Dfc_{\Delta_0}(R(K, n), R'(K, n)),
\end{multline*}
holds for every $\Delta_0 > 0$.
\end{lem}

\begin{rem*}
Definition~\ref{def:defect_prop} involves also the constant $C_{cnt}$ (needed for the Point Counting property). However, this constant is not relevant for the Thin Box Lemma.
We will need the Point Counting property and the constant $C_{cnt}$ for the further steps.
\end{rem*}

\subsection{Reduction to existence of a repair algorithm}

First of all, we wish to emphasize the set of input parameters for the Thin Box Lemma. This set consists of the constants $d, \varepsilon, \rho, c$ and the function $\Delta$.
From this point and until the end of the section we assume all these parameters to be fixed.

Let $\zeta \in \Omega(\mathbb R^d)$, $K > 100 \rho$, and let $n$ be a positive integer. Assume $\phi : \Omega^{(s)}(R(K, n), \zeta) \to \Omega(\mathbb R^d)$ is a $(R'(K, n), \rho)$-saturator.
In addition, assume that a constant $\Delta_0 > 0$ satisfies
$$\Delta_0 > \sup\limits_{\xi \in \Omega^{(s)}(R(K, n), \zeta)} \Delta(\xi, \phi(\xi), R'(K, n)).$$
We aim to construct an algorithm as in the following Definition~\ref{defn:rep_alg}.

\begin{defn}[Repair algorithm]\label{defn:rep_alg}
Let $\mathcal A$ be an algorithm with the following structure: \newline
{\bf Input:} $\zeta, \tilde{K}, n, \phi, \Delta_0$ as above; $\xi \in \Omega^{(s)}(R(\tilde{K}, n), \zeta)$. \newline
{\bf Output:} A sequence $(\xi_1, \xi_2, \ldots, \xi_k)$, where $\xi_1 = \xi$, $\xi_i \in \Omega^{(s)}(R(\tilde{K}, n), \zeta)$, and the length $k$ depends on the input. \newline
For each fixed $\mathcal I = (\zeta, \tilde{K}, n, \phi, \Delta_0)$ denote by $\Xi(\mathcal I)$ the set of all possible output sequences of $\mathcal A$. Further, denote
$$Z_i(\mathcal I) = \{ \xi \in \Omega^{(s)}(R(\tilde{K}, n), \zeta) : \text{$\exists (\xi_1, \xi_2, \ldots, \xi_k) \in \Xi(\mathcal I)$ such that $k \geq i$ and $\xi_i = \xi$} \},$$
$$Z^{term}_i(\mathcal I) = \{ \xi \in \Omega^{(s)}(R(\tilde{K}, n), \zeta) : \text{$\exists (\xi_1, \xi_2, \ldots, \xi_i) \in \Xi(\mathcal I)$ such that $\xi_i = \xi$} \}.$$
We say that $\mathcal A$ {\it is a repair algorithm} if there exists a positive constant $K > 100 \rho$, and positive-valued functions $k_0 = k_0(\Delta_0)$ and $c_0 = c_0(\Delta_0)$
such that in the case $\mathcal I = (\zeta, K, n, \phi, \Delta_0)$ one necessarily has
\begin{enumerate}
  \item Every sequence $(\xi_1, \xi_2, \ldots, \xi_k) \in \Xi(\mathcal I)$ (deterministically) satisfies $k \leq k_0$.
  \item If $\xi \in Z^{term}_i(\mathcal I)$ then at least one of the following holds:
        \begin{itemize}
          \item $R'(K, n)$ is $\left( \varepsilon, \frac{6 \Delta_0}{c} \right)$-crossed by $\xi$,
          \item $\xi$ admits an empty $\varepsilon$-space in $R(K, n)$.
        \end{itemize}
  \item There exists a positive constant $c_0$ such that the inequality
        $$\Pr (\eta \in Z_{i + 1}(\mathcal I)) \geq c_0 \Pr (\eta \in Z_i(\mathcal I) \setminus Z^{term}_i(\mathcal I))$$
        holds with $\eta = \eta^{(s)}(R(K, n), \zeta)$
\end{enumerate}
\end{defn}

Now we make the key reduction of the Thin Box Lemma.

\begin{lem}
If there exists a repair algorithm then the Thin Box Lemma holds.
\end{lem}

\begin{proof}
With no loss of generality, let $c_0 \in (0, 1)$. Further, we assume that $\mathcal I = (\zeta, K, n, \phi, \Delta_0)$
is fixed, therefore we write simply $Z_i$ and $Z^{term}_i$ omitting the argument $\mathcal I$. Finally, we write
$\eta = \eta^{(s)}(R(K, n), \zeta)$ as in the definition of a repair algorithm.

We will prove that
\begin{equation}\label{eq:inf}
p_{cross}\left[ \varepsilon, \frac{6 \Delta_0}{c} \right](K, n, \zeta, s) + p_{empty}[\varepsilon](K, n, \zeta, s) \geq
\frac{1}{2}\left(\frac{c_0}{2} \right)^{k_0}.
\end{equation}
We argue by contradiction: assume that~\eqref{eq:inf} is false.

Let us show, by induction, that
$$\Pr(\eta \in Z_i) \geq \left(\frac{c_0}{2} \right)^i.$$
For $i = 1$ this is immediate, since $Z_0 = \Omega^{(s)}(R(K, n), \zeta)$.

Assume that the inequality is proved up to some $i$. One can observe that
\begin{multline*}
\frac{1}{2}\left(\frac{c_0}{2} \right)^i \geq
\frac{1}{2}\left(\frac{c_0}{2} \right)^{k_0} \geq \\
p_{cross}\left[ \varepsilon, \frac{6 \Delta_0}{c} \right](K, n, \zeta, s) + p_{empty}[\varepsilon](K, n, \zeta, s) \geq \Pr(\eta \in Z^{term}_i).
\end{multline*}
Indeed, the second inequality is exactly the contrary to~\eqref{eq:inf}, while the third inequality follows from
the Key Property 2. Consequently,
$$\Pr(\eta \in Z_i \setminus Z^{term}_i) \geq \frac{1}{2}\left(\frac{c_0}{2} \right)^i.$$
By Key Property 3, one concludes
$$\Pr(\eta \in Z_{i + 1}) \geq \left(\frac{c_0}{2} \right)^{i + 1}.$$
The induction step is verified.

By Key Property 1, $Z_{k_0} = Z^{term}_{k_0}$. Hence
\begin{multline*}
p_{cross}\left[ \varepsilon, \frac{6 \Delta_0}{c} \right](K, n, \zeta, s) + p_{empty}[\varepsilon](K, n, \zeta, s) \geq \\
\Pr(\eta \in Z^{term}_{k_0}) = \Pr(\eta \in Z_{k_0}) \geq
\left(\frac{c_0}{2} \right)^{k_0}.
\end{multline*}
This contradicts our assumption that~\eqref{eq:inf} is false.
\end{proof}

\subsection{Elementary moves}

Consider the line
$$\ell = \{ 0 \}^{d - 1} \times \mathbb R.$$
Let $K > 0$. Define two real-valued functions
$$h^K_-, h^K_+ : \{ x : \dist(x, \ell) < K \} \to \mathbb R$$
as follows:
$$B_K(x) \cap \ell = \{ 0 \}^{d - 1} \times [h^K_-, h^K_+] + t.$$
Accordingly, let us write
$$x^K_{\pm} = \{ 0 \}^{d - 1} \times \{ h^K_{\pm} \} + t.$$
Thus $x^K_{\pm}$ are two points on the line $\ell$ at distance $K$ from $x$. In the construction below we refer to $K$
as to the {\it radius of the elementary move}.

\begin{defn}
Let $\xi \in \Omega(\mathbb R^d)$, $x \in \xi$ and $dist(x, \ell) < K$. Let $x' \in [x, x^K_-]$ be a point satisfying $\| x -  x' \| = m\frac{\varepsilon}{10}$, where
$m \in \left\{1, \ldots, \left\lceil \frac{100}{\varepsilon} \right\rceil \right\}$. Then
\begin{enumerate}
  \item If $\{y \in \xi : \text{$y \neq x$ and $\| y - x' \| \leq 2$} \} = \varnothing$, we will say that set $(\xi \setminus \{ x \}) \cup \{ x' \}$ is {\it obtained from $\xi$ by an elementary move of $x$ with magnitude $m$}.
  \item Otherwise we say that an elementary move of $x$ to $x'$ {\it is forbidden} by any point $y \in (\xi \setminus \{ x \}) \cap B_2(x')$.
\end{enumerate}
\end{defn}

We proceed by proving two crucial properties of elementary moves.

\begin{lem}\label{lem:move}
There exists a positive constant $K_1$ such that for every $K > K_1$ the following holds.
If $\xi \in \Omega(\mathbb R^d)$, $x \in \xi$ and a point $y \in \xi$ forbids some elementary move of $x$ with radius $K$,
then $\| y - x^K_- \| < \rho$.
\end{lem}

\begin{proof}
Let $x'$ be the point such that the move of $x$ to $x'$ be forbidden by $y$. Then $\| x -  x' \| = m\frac{\varepsilon}{10}$, where $m$ an integer,
$1 \leq m \leq \left\lceil \frac{100}{\varepsilon} \right\rceil$. Also, with no loss of generality assume that $x^K_-$ coincides with the origin $\mathbf 0$.

Since $x, y \in \xi$ and $\zeta \in \Omega(\mathbb R^d)$, we have $\| x - y \| > 2 \geq \| x' - y \|$.
Thus the perpendicular bisector to the segment $[x, x']$ separates $x$ from $x'$ and $y$. Consequently,
$$\langle x - x', x \rangle + \langle x - x', x' \rangle > 2 \langle x, y \rangle.$$
Then
$$\langle x, x \rangle + \langle x, x' \rangle > 2 \langle x, y \rangle,$$
since $x = \frac{10K}{m \varepsilon} (x - x')$.
After subtracting $2\langle x, x \rangle$ from each side and reversing the sign, one obtains
$$\langle x, x - x' \rangle < 2 \langle x, x - y \rangle.$$
But $m \geq 1$, therefore $\langle x, x - x' \rangle \geq \frac{\varepsilon}{10} K$. Hence
$$\langle x, x - y \rangle \geq \rho \frac{\varepsilon}{20}.$$
The above implies
\begin{multline*}
\langle y, y \rangle = \langle x - (x - y), x - (x - y) \rangle = \langle x, x \rangle - 2 \langle x, x - y \rangle + \langle x - y, x - y \rangle \leq \\
K^2 - K \frac{\varepsilon}{10} + \langle x - y, x - y \rangle.
\end{multline*}
But
$$\langle x - y, x - y \rangle = \| x - y \|^2 \leq (\| x - x' \| + \| x' - y \|)^2 \leq (20 + 2)^2 < 500.$$
Hence with $K_1 = \frac{5000}{\varepsilon}$ one concludes
$$\| y - x^K_- \|^2 = \langle y, y \rangle < K^2.$$
This finishes the proof. \end{proof}

\begin{lem}\label{lem:dist}
There exists a positive constant $K_2$ such that for every $K > K_2$ the following holds.
Assume that for $x, y \in \mathbb R^d$ the conditions below are satisfied:
\begin{enumerate}
  \item $\dist(x, \ell) < K$, $\dist(y, \ell) < K$.
  \item $\| x - y \| < 10$.
  \item $x' \in [x, x^{\rho}_-]$, $y' \in [y, y^{\rho}_-]$.
  \item $\| x - x' \| = \| y - y' \| < 20$.
\end{enumerate}
Then $-\frac{\varepsilon}{10} < \| x - y \| - \| x' - y' \| < \frac{\varepsilon}{10}$.
\end{lem}

\begin{proof}

Let $b = \| x - x' \|$. Then
$$x' = \frac{b}{K} x^K_- + \frac{K - b}{K} x, \quad y' = \frac{b}{K} y^K_- + \frac{K - b}{K} y.$$
Therefore
$$ x' - y' = (x - y) - \frac{b}{K}(x - y) + \frac{b}{K} (x^K_- - y^K_-).$$
Hence the conclusion of lemma will hold if the following inequalities are satisfied:
\begin{eqnarray}
    \| x - y \| < \frac{K}{20} \cdot \frac{\varepsilon}{20}, \label{eq:moves:4}\\
    \| x^K_- - y^K_- \| < \frac{K}{20} \cdot \frac{\varepsilon}{20}. \label{eq:moves:5}
\end{eqnarray}

Our aim is to show that $K_2 = \left( \frac{2000}{\varepsilon} \right)^2$ is sufficient to satisfy both~\eqref{eq:moves:4} and~\eqref{eq:moves:5}.
In particular, the restriction on $K_2$ implies $K_2 \geq \frac{8000}{\varepsilon}$, so~\eqref{eq:moves:4} is immediate.

Let us turn to the inequality~\eqref{eq:moves:5}. Without loss of generality, assume that $h^K_-(x) \geq h^K_-(y)$.
Then the angle between the vectors $x - x^K_-$ and  $y^K_- - x^K_-$ is right or obtuse, therefore
$$\| x^K_- - y^K_- \|^2 \leq \| x - y^K_- \|^2 - \| x - x^K_- \|^2 \leq (K + 10)^2 - K^2 = 20K + 100.$$
From the assumption on $K_2$, we have, in particular, $K > 20$, and, consequently, $20K + 100 < 25K$. Thus
$$\| x^K_- - y^K_- \| < 5\sqrt{K} = \frac{K}{20} \cdot \frac{100}{\sqrt{K}} < \frac{K}{20} \cdot \frac{\varepsilon}{20},$$
and~\eqref{eq:moves:5} is proved as well, completing the proof of the lemma. \end{proof}

\subsection{Implementation of the repair algorithm}

We are ready to implement an algorithm which will satisfy the definition of a repair algorithm.

Let $\mathcal I = (\zeta, K, n, \phi, \Delta_0)$ and $\xi \in \Omega^{(s)}(R(K, n), \zeta)$ be the input.
For simplicity, we will write $P_i$, $R$ and $R'$ instead of $P_i(K)$, $R(K, n)$ and $R'(K, n)$, respectively.
We proceed as follows.

\noindent{\bf 1. Classification of the cubes $P_i$.} We attribute each $P_i$ ($-2n + 1 \leq i \leq 2n - 1$) to one of the five types according to the rule below.
\begin{itemize}
  \item If $\Delta(\xi, \phi(\xi), P_i) \geq c/2$, consider the sets
        $$J_-(i) = \{ j \in \mathbb Z : \text{$-2n + 1 \leq j < i$ and $\Delta(\xi, \phi(\xi), P_j) < c/2$} \},$$
        $$J_+(i) = \{ j \in \mathbb Z : \text{$-2n + 1 \leq j < i$ and $\Delta(\xi, \phi(\xi), P_j) < c/2$} \}.$$
        $P_i$ is attributed to {\bf Type A} if neither of the sets $J_-(i)$ and $J_+(i)$ is empty. Otherwise $P_i$ is attributed to {\bf Type B}.
  \item $P_i$ is attributed to {\bf Type C} if $-2n + 2 \leq i \leq 2n - 2$ and $\Delta(\xi, \phi(\xi), P_j) < c/2$ for each $j \in \{ i - 1, i, i + 1 \}$.
  \item $P_i$ is attributed to {\bf Type D} if $\Delta(\xi, \phi(\xi), P_j) < c/2$ and there exists $j \in \{ i - 1, i + 1\}$ such that $P_j$ is attributed to Type A.
  \item $P_i$ is attributed to {\bf Type E} if neither of the above applies.
\end{itemize}

\noindent{\bf 2. Auxiliary routine: processing a maximal sequence of neighboring Type A cubes.} Let the indices $j_1, j_2$ ($-2n + 1 \leq j_1 < j_2 \leq 2n - 1$)
satisfy the property: $P_{j_1}$ and $P_{j_2}$ are two consecutive cubes of Type D (i.e., no $P_i$ with $j_1 < i < j_2$ is of Type D). Then the set
$\{ P_{j_1 + 1}, P_{j_1 + 2}, \ldots, P_{j_2 - 1} \}$ consists only of Type A cubes. Additionally assume $j_1 + 1 < j_2$ so that the sequence of Type A cubes
separating $P_{j_1}$ from $P_{j_2}$ is non-empty.

Let the point $o_i$ denote the center of the cube $P_i$. Define
$$w = \argmin\limits_{x \in \xi \cap B_K(o_{j_2})} h^K_- (x).$$
Let
$$\{ x : x \in \xi, \dist(x, \ell) < K, 10j_1K \leq h^K_-(x) \leq h^K_-(w) \} = \{ x_1, x_2, \ldots, x_q \},$$
where the points $x_i$ are sorted in increasing order with respect to the function $h^K_-$ (i.e., the value $h^K_-(x_i)$ increases with $i$).

The routine runs consecutively over $i = 1, 2, \ldots, q$ and finds a new position $x'_i$ for each $x_i$ as follows.
We say that an elementary move {\it connects} $x$ and $y$ if $\| x' - y \| < 2 + \varepsilon$, where $x'$ is the new position of $x$.
\begin{itemize}
  \item Check if $\dist(x_i, (\xi \cap B_K(o_{j_1})) \cup \{x'_1, x'_2, \ldots, x'_{i - 1} \}) \leq 2 + \varepsilon$.
        If yes, set $x'_i = x_i$, replace $i$ by $i + 1$ and start over. If no, proceed to the next step.
  \item Check if there exists a (non-forbidden) elementary move of $x_i$ with radius $K$ and connecting $x_i$ to some point
        $y \in (\xi \cap B_K(o_{j_1})) \cup \{x'_1, x'_2, \ldots, x'_{i - 1} \}$. If yes, choose the smallest possible magnitude for such a move.
        Define $x'_i$ to be the result of the elementary move chosen. Replace $i$ by $i + 1$ and start over. If no, proceed to the next step.
  \item Retain all $x_j$ ($j \geq i$) in place. Report this event and terminate the routine.
\end{itemize}

By our construction, it is clear that
$$(\xi \setminus \{ x_1, x_2, \ldots, x_i \} ) \cup \{ x'_1, x'_2, \ldots, x'_i \} \in \Omega(\mathbb R^d).$$

\noindent{\bf 3. Using the routine of Step 2.} We apply the routine of Step 2 to each maximal set of consequent Type A cubes.
We will show that the processes do not interact (see Lemma~\ref{lem:routine}, Assertion 1), so we can perform the routines consecutively from the lowermost
sequence of neighboring Type A cubes to the uppermost one. If some instance of the routine reports a termination, we terminate the entire algorithm.

For the rest of the section we denote the above defined algorithm by $\mathcal A$. The sequence $(\xi_1, \xi_2, \ldots, \xi_k)$ is, by definition,
the sequence of configurations obtained from $\xi = \xi_1$ after each elementary move of $\mathcal A$.

\subsection{Auxiliary lemmas on the algorithm}

Let us prove some immediate properties of the above defined algorithm. Denote
$$I_A = \{ i : \text{the cube $P_i$ is of Type A} \}.$$
$I_B$, $I_C$, $I_D$ and $I_E$ are defined similarly.

\begin{lem}\label{lem:card_i}
$\# (I_A \cup I_B \cup I_D \cup I_E) \leq \frac{6 \Delta_0}{c}$.
\end{lem}

\begin{proof}
By Monotonicity and Additivity properties of the defect function, one has
$$\Delta_0 \geq \Delta(\xi, \phi(\xi), R') \geq \sum\limits_{i \in I_A \cup I_B} \Delta(\xi, \phi(\xi), P_i) \geq
\# (I_A \cup I_B) \cdot \frac{c}{2}.$$
Hence $\# (I_A \cup I_B) \leq \frac{2\Delta_0}{c}$.

Further, $\# (I_C \cup I_E) \leq 2 \# (I_A \cup I_B)$, because each Type C or Type E cube is adjacent to some Type A or Type B cube. Consequently,
$$\# (I_A \cup I_B \cup I_C \cup I_E) \leq 3 \# (I_A \cup I_B) \leq \frac{6\Delta_0}{c}.$$
\end{proof}

\begin{lem}\label{lem:routine}
Let $K = \max(K_1, K_2, 100\rho, 1000 / \varepsilon) + 1$. Assume that a sequence of neighboring Type A cubes $\{ P_{j_1 + 1}, P_{j_1 + 2}, \ldots, P_{j_2 - 1} \}$ is
subject to the above defined auxiliary routine. Then the following assertions are true:
\begin{enumerate}
  \item If a point $x \in \xi$ is affected by the routine then
        \begin{eqnarray*}
          \dist(x, [o_{j_1}, o_{j_2}]) \leq K, \\
          \|x - o_{j_1}\| \geq K, \\
          \|x - o_{j_2}\| \geq K.
        \end{eqnarray*}
  \item If the routine terminates while attempting to move $x_i$, and
	    $$\tilde{\xi} = \xi \setminus \{x_1, x_2, \ldots, x_{i - 1} \} \cup \{x'_1, x'_2, \ldots, x'_{i - 1} \}$$
		then $\tilde{\xi}$ has an $\varepsilon$-empty space in $R$.
\end{enumerate}
\end{lem}

\begin{proof}
Assertion 1 is clear from the construction.

Let us prove Assertion 2. Given $i$ as in the assertion, denote
\begin{multline*}
J = \biggl\{ m \in \mathbb Z, \text{$0 \leq m \leq \left\lceil \frac{100}{\varepsilon} \right\rceil$ and} \\
\dist \left(\frac{m \frac{\varepsilon}{10}}{K} \cdot (x_i)^K_- + \frac{K - m \frac{\varepsilon}{10}}{K} \cdot x_i, (\xi \cap B_K(o_{j_1})) \cup \{x'_1, x'_2, \ldots, x'_{i - 1} \} \right) \leq 2 + \varepsilon \biggr\}.
\end{multline*}

Consider three cases.

\noindent {\bf Case 1.} $0 \in J$. Then $\dist(x_i, (\xi \cap B_K(o_{j_1})) \cup \{x'_1, x'_2, \ldots, x'_{i - 1} \}) \leq 2 + \varepsilon$. Therefore the routine
sets $x'_i = x_i$ and does not terminate. This contradicts our initial assumption on $i$, so this case is impossible.

\noindent {\bf Case 2.} $J \neq \varnothing$ and $0 \notin J$. Let $m = \min J$. We claim that the routine sets $x'_i = z$, where
$$z = \frac{m \frac{\varepsilon}{10}}{K} \cdot (x_i)^K_- + \frac{K - m \frac{\varepsilon}{10}}{K} \cdot x_i$$
and does not terminate. In order to prove the claim, we will check that the elementary move of $x_i$ to $z$ as above is not forbidden.
Assume the converse: there is a point $y \in (\xi \setminus \{ x_1, x_2, \ldots, x_{i - 1} \} ) \cup \{ x'_1, x'_2, \ldots, x'_{i - 1} \} $
satisfying $\| y - z \| < 2$. There are five subcases.
\begin{enumerate}
	\item $y = x'_j$, where $j < i$. But
          $$ \left\| z - \left( \frac{(m - 1) \frac{\varepsilon}{10}}{K} \cdot (x_i)^K_- + \frac{K - (m - 1) \frac{\varepsilon}{10}}{K} \cdot x_i \right)  \right \| = \frac{\varepsilon}{10},$$
          thus $m - 1 \in J$. This contradicts the assumption $m = \min J$, hence the subcase is impossible.
	\item $y = x_j$, where $j > i$. Lemma~\ref{lem:move} immediately implies
	      $h^K_-(x_j) = h^K_-(y) < h^K_-(x_i)$. But this is impossible, since $x_1, x_2, \ldots, x_q$ are arranged in increasing
          order with respect to the function $h^K_-$. This subcase is also impossible.
    \item $\dist (y, \ell) \geq K$. But, Lemma~\ref{lem:move} implies $\dist (y, \ell) \leq \| y - (x_i)^K_- \| < K$. Hence this subcase is impossible.
    \item $y \in \xi \setminus \{ x_1, x_2, \ldots, x_q \}$ and $h^K_-(y) > 10j_1K$. On the other hand, $h^K_-(y) < h^K_-(x_i)$ by Lemma~\ref{lem:move}.
          Then $h^K_-(y) < h^K_-(x_i) < h^K_-(w)$. Consequently, $y = x_j$, $j < i$. This contradicts the definition of the subcase, therefore this subcase is impossible.
    \item $y \in \xi \setminus \{ x_1, x_2, \ldots, x_q \}$ and $h^K_-(y) \leq 10j_1K$.$h^K_-(y) \leq 10j_1K$. By Lemma~\ref{lem:move}, we also have $h^K_+(y) > h^K_-(x_i) \geq 10j_1K$.
          Consequently, $o_{j_1} \in [y^K_-, y^K_+]$ and thus $y \in \xi \cap B_K(o_{j_1})$. Similarly to Subcase 1 we conclude $m - 1 \in J$.
          This again contradicts the assumption $m = \min J$, hence the subcase is impossible.
\end{enumerate}
All the subcases are impossible, consequently, the routine does not terminate as claimed.

\noindent {\bf Case 3.} $J = \varnothing$. Choose a point $z \in \mathbb R^d$ as follows:
$$y \in [x_i, (x_i)^K_-], \qquad \| z - x_i \| = \left\lceil \frac{50}{\varepsilon} \right\rceil \frac{\varepsilon}{10}.$$
Then $\left\lceil \frac{50}{\varepsilon} \right\rceil \notin J$ implies
$B_{2 + \varepsilon}(z) \cap \tilde{\xi} = \varnothing$. In addition,
$$B_{2 + \varepsilon}(z) \subset B_3(z) \subset B_K((x_i)^K_-) \subset R.$$
Hence $\tilde{\xi}$ indeed has an $\varepsilon$-empty space. \end{proof}

\subsection{$\mathcal A$ is a repair algorithm}

We proceed by verifying the properties 1--3 of Definition~\ref{defn:rep_alg} for the algorithm constructed above.
We need to check these properties for some particular value of $K$, therefore we fix $K = \max(K_1, K_2, 100\rho, 1000 / \varepsilon) + 1$.
Each property will be verified in a separate lemma.

\begin{lem}
The algorithm $\mathcal A$ satisfies property 1. Equivalently, the sequence $(\xi_1, \xi_2, \ldots, \xi_k)$ produced by the algorithm satisfies
$k \leq k_0(\Delta_0)$.
\end{lem}

\begin{proof}
Each point of the set $\xi$ is moved at most once. Therefore the number of points affected by $\mathcal A$ equals $k - 1$.
By Lemma~\ref{lem:routine}, Assertion 1, whenever a point $x \in \xi \cap P_i$ is affected by $\mathcal A$, one necessarily
has $i \in I_A \cup I_D$.

On the other hand, the inequality $\# (\zeta \cap P_i) \leq \mathrm{const}$, since the edges of the cube $P_i$ have fixed
length $k$. Hence, using Lemma~\ref{lem:card_i}, one obtains
\begin{equation*}
k \leq 1 + \# (I_A + I_D) \cdot \mathrm{const} = 1 + \mathrm{const} \cdot \frac{6\Delta_0}{c}. \qedhere
\end{equation*}
\end{proof}

\begin{lem}
The algorithm $\mathcal A$ satisfies property 2. Equivalently, if $(\xi_1, \xi_2, \ldots, \xi_k)$ is an output of
$\mathcal A$ and $\xi_k$ does not admit an empty $\varepsilon$-space in $R$, then or $R'$ is $\left(\varepsilon, \frac{6\Delta_0}{c} \right)$-crossed by $\xi_k$.
\end{lem}

\begin{proof}
Assume $\xi_k$ indeed does not admit an empty $\varepsilon$-space in $R$. Then by Lemma~\ref{lem:routine}, Assertion 2, none of the auxiliary routines
performed by $\mathcal A$ reports termination.

Our proof that $R'$ is $\left(\varepsilon, \frac{6\Delta_0}{c} \right)$-crossed by $\xi_k$ will require verification of the three properties
(see Definition~\ref{def:cross}). These properties, which we call {\it Crossing Properties}, involve the set of indices $I$, which we take equal to $I_C$.

By Lemma~\ref{lem:card_i} we have $\# I \geq 4n - 1 - \frac{6\Delta_0}{c}$. Therefore Crossing Property 1 holds.

Let an index $i \in I$ be given. Assume the points $x, x' \in \zeta_k \cap P_i$ satisfy $B_{\rho}(x), B_{\rho}(x') \subset P_i$. The definition of $I = I_C$ implies $\Delta(\xi, \phi(\xi), P_i) < c$.
By the Connectivity property of $\Delta$ one concludes that $x$ and $x'$ belong to the same connected component of $G_{\varepsilon}(\xi \cap P_i)$. But Lemma~\ref{lem:routine}, Assertion 1,
implies $G_{\varepsilon}(\xi \cap P_i) = G_{\varepsilon}(\xi_k \cap P_i)$, hence Crossing Property 2 holds.

Before verifying Crossing Property 3, we establish several auxiliary facts.

\noindent {\bf A.} Let $i, i + 1 \in I_C \cup I_D$,
$$z_1 \in \xi_k \cap B_2(o_i), \quad z_2 \in \xi_k \cap B_2(o_{i + 1}).$$
Then $z_1$ and $z_2$ belong to the same connected component of $G_{\varepsilon}(\xi_k \cap R')$.

Lemma~\ref{lem:routine}, Assertion 1 yields $z_1, z_2 \in \xi$. Now denote $D = \conv (B_K(o_i) \cup B_K(o_{i + 1}))$. Then
$$\Delta(\xi, \phi(\xi), D) \leq \Delta(\xi, \phi(\xi), P_i) + \Delta(\xi, \phi(\xi), P_{i + 1}) < c.$$
By the Connectivity property of $\Delta$, the points $z_1$ and $z_2$ belong to the same connected component of $G_{\varepsilon}(\xi \cap D)$.
But Lemma~\ref{lem:routine}, Assertion 1 implies $\xi_k \cap D \supseteq \xi \cap D$, hence {\bf A} follows.

\noindent {\bf B.} Let $j_1, j_2 \in I_C$, $j_1 < j_2 - 1$ and $i \in I_A$ for all $i$ satisfying $j_1 < i < j_2$. Let
$$z_1 \in \xi_k \cap B_2(o_{j_1}), \quad z_2 \in \xi_k \cap B_2(o_{j_2}).$$
Then $z_1$ and $z_2$ belong to the same connected component of $G_{\varepsilon}(\xi_k \cap R')$.

By the conditions of {\bf B}, a single auxiliary routine of $\mathcal A$ affects the set
$$\xi \cap \bigcup_{i = j_1}^{j_2} P_i.$$
Let $x_1, x_2, \ldots, x_q$ and $w$ be as in the definition of the routine. Since the routine does not report termination,
$x'_i \in \zeta_k$ is defined for every $i = 1, 2, \ldots, q$. We continue the proof of {\bf B} in several steps.

\noindent {\bf B.1)} For each $i = 1, 2, \ldots, q$ there is $u_i \in \xi \cap B_K(o_{j_1})$ such that
$x'_i$ and $u_i$ belong to the same connected component of $G_{\varepsilon}(\xi_k \cap R')$.

We argue by induction over $i$. For $i = 1$ the statement is clear from the construction of $x'_1$. If $i > 1$
and $\| x'_i - x'_j \| < 2 + \varepsilon$ for some $j < i$ then we apply the induction hypothesis for $x'_j$. Otherwise
$\dist(x'_i, \xi \cap B_K(o_{j_1})) \leq 2 + \varepsilon$ and {\bf B.1)} follows immediately.

\noindent {\bf B.2)} There exists an index $i_0 \in \{ 1, 2, \ldots, q \}$ such that $\| w - x_{i_0} \| \leq 2 + \varepsilon$
and $x'_{i_0} = x_{i_0}$.

Define $x_{i_0}$ to be the point of the set $\xi$ with the following properties:
$$\| x_{i_0} - w \| < 2 + \frac{\varepsilon}{2}, \quad h^K_-(x_{i_0}) < h^K_-(w).$$
There exists at least one such point because of the Distance Decrement property of $\Delta$ applied to the domain $B_K(w^K_-)$.
We proceed by defining a sequence of indices $i_0, i_1, \ldots, i_{\lceil 100 / \varepsilon \rceil}$ recursively by fulfilling the condition
$$\| x_{i_{j + 1}} - x_{i_j} \| < 2 + \frac{\varepsilon}{2}, \quad h^K_-(x_{i_{j + 1}}) < h^K_-(x_{i_j}).$$
(Since $K > 1000 / \varepsilon$, one has $B_K((x_{i_j})^K_-) \subset P_{j_2}$, therefore the Distance Decrement property is applicable.)

Let $m_j$ be the magnitude of the elementary move applied to $x_{i_j}$ (in the case $x'_{i_j} = x_{i_j}$ we set $m_j = 0$ by definition).
We prove that $m_j \leq j$. The proof is inductive, from $j = \lceil 100 / \varepsilon \rceil$ to $j = 0$.

Indeed, the inequality $m_{\lceil 100 / \varepsilon \rceil} \leq \lceil 100 / \varepsilon \rceil$ clearly holds. Now let
$j < \lceil 100 / \varepsilon \rceil$. If $m_{j + 1} = 0$, then $m_j = 0$ and there is nothing to prove. Otherwise denote
$$y = \frac{(m_{j + 1} - 1) \frac{\varepsilon}{10}}{K} \cdot (x_{i_j})^K_- + \frac{K - (m_{j + 1} - 1) \frac{\varepsilon}{10}}{K} \cdot x_{i_j}.$$
Using Lemma~\ref{lem:dist}, one concludes
$$\| y - x'_{i_{j + 1}} \| < \| x_{i_j} - x_{i_{j + 1}} \| + 2 \cdot \varepsilon / 10 \leq 2 + 0.7 \varepsilon.$$
Therefore $m_j \leq m_{j + 1} - 1$, hence the induction step follows.

\noindent {\bf B.3)} There is a point $w' \in \xi \cap B_{K - \rho}(o_{j_2})$ such that $w$ and $w'$ belong to the same connected
component of $G_{\varepsilon}(\xi_k \cap R')$.

The path $y_0 = w, y_1, y_2, \ldots, y_l = w'$ can be constructed recursively: $y_{i + 1}$ is obtained by applying the Distance Decrement property
of $\Delta$ to the ball $B_{\| y - o_{j_2} \|}(o_{j_2})$.

\noindent {\bf B.4)} Let $i_0$ be as in {\bf B.2)} and $u_{i_0}$ be as in {\bf B.1)}.  Then there is a point $u' \in \xi \cap B_{K - \rho}(o_{j_1})$
such that $u_{i_0}$ and $u'$ belong to the same connected component of $G_{\varepsilon}(\xi_k \cap R')$.

If $\| u_{i_0} - o_{j_1} \| \leq K - \rho$, there is nothing to prove. Otherwise one argues as in {\bf B.3)}.

\noindent {\bf B.5)} The points $z_1$ and $z_2$ belong to the same connected component of $G_{\varepsilon}(\xi_k \cap R')$.

Let $\mathcal C$ be the connected component of $G_{\varepsilon}(\xi_k \cap R')$ containing $z_2$ and let $V$ be the vertex set of $\mathcal C$.

By the Connectivity property of $\Delta$ applied to $B_K(o_{j_2})$, $v \in V$. By {\bf B.3)}, $w \in V$. By~{\bf B.2)}, $x_{i_0} \in V$.
By {\bf B.1)}, $u_{i_0} \in V$. By {\bf B.4)}, $u' \in V$. Finally, by the Connectivity property of $\Delta$ applied to $B_K(o_{j_1})$, $z_1 \in V$.
Thus {\bf B} is verified.

\noindent {\bf C.} Let $i \in I_C$,
$$z \in \xi_k \cap B_2(o_i), \quad z' \in \xi_k, \quad B_{\rho}(z') \subset P_i.$$
Then $z$ and $z'$ belong the same connected component of $G_{\varepsilon}(\xi_k \cap R')$.

This is a direct consequence of the Connectivity Property of $\Delta$ applied to $P_i$.

\noindent {\bf D.} Let $i \in I_C \cup I_D$. Then $\xi_k \cap B_2(o_i) \neq \varnothing$.

This is a direct consequence of Lemma~\ref{lem:routine}, Assertion 1, and the Saturation property of $\Delta$ applied to $P_i$.

Let us return to Crossing Property 3. It is sufficient to show that the two points $z_1, z_2 \in \xi_k$ belong to the same
connected component of $G_{\varepsilon}(\zeta_k \cap R')$ once the following conditions are satisfied:
\begin{eqnarray*}
  z_1 \in P_{i_1}, & \quad z_2 \in P_{i_2}, & \qquad \text{where $i_1 \neq i_2$}  \\
  B_{\rho}(z_1) \subset P_{i_1}, & \quad B_{\rho}(z_1) \subset P_{i_1} &
\end{eqnarray*}
But this is an immediate consequence of the facts {\bf A} -- {\bf D}. \end{proof}

\begin{lem}
The algorithm $\mathcal A$ satisfies property 2. Equivalently, there exists a positive constant $c_0 = c_0(\Delta_0)$
such that the following inequality holds for the random configuration $\eta = \eta^{(s)}(R, \zeta)$:
\begin{equation}\label{eq:shrink}
\Pr (\eta \in Z_{i + 1}) \geq c_0 \Pr (\eta \in Z_i \setminus Z^{term}_i).
\end{equation}
\end{lem}

\begin{proof}
For each $\zeta_{i + 1} \in Z_{i + 1}$ consider all its possible predecessors $\zeta_i \in Z_i$. This is a multi-valued
map $\mathfrak f$ from $Z^{i + 1}$ to $Z_i$. The set $\mathfrak(Z_{i + 1})$ of all values attained by $\mathfrak f$
satisfies $Z_i \setminus Z^{term}_i \subset \mathfrak(Z_{i + 1})$.

We will verify the following inequalities.
\begin{enumerate}
  \item[\bf a.] For any sheet $f$ of the multi-map $\mathfrak f$ and any $\xi \in Z_{i + 1}$ such that
                $f$ is defined in a neighborhood of $\xi$ one has
                $$|Jac\, f(\xi)| < \left( \frac{K}{K - 20} \right)^{d - 1}.$$
  \item[\bf b.] For a.e. $\xi \in Z_{i + 1}$ one has
                $$\# \mathfrak f(\xi) < \left\lceil \frac{100}{\varepsilon} \right\rceil \cdot \frac{6 \Delta_0}{c} \cdot \frac{(10K + 2)^d}{|B_1(\mathbf 0)|}.$$
\end{enumerate}

\noindent {\bf Proof of a.} Consider the map $g_m: \bigcup\limits_{y \in \ell} B_K(y) \to \mathbb R^d$ defined by
$$g_m(x) = \frac{m \frac{\varepsilon}{10}}{K} \cdot x^K_- + \frac{K - m \frac{\varepsilon}{10}}{K} \cdot x.$$
One can check that
$$|Jac\, g_m(x)| = \left( \frac{K - m\varepsilon / 10}{K} \right)^{d - 1}.$$

If $f$ and $\xi$ are as above then there exists an integer $m \in \{1, 2, \ldots, \lceil 100 / \varepsilon \rceil \}$ and a point $x_0 \in \xi$
such that the identity
\begin{equation}\label{eq:sheet}
f(\xi') = (\xi' \setminus \{ x \}) \cup g_m^{-1}(x), \quad \text{where $\{ x \} = \xi' \cap B_1(x_0)$,}
\end{equation}
holds for any $\xi'$ in a sufficiently small neighborhood of $\xi$. In other words, $f$ affects a single point of $\xi'$ by an inverse of an
elementary move. Thus
$$|Jac\, f(\xi)| = |Jac\, g_m(x_0)|^{-1} < \left( \frac{K}{K - 20} \right)^{d - 1}.$$

\noindent {\bf Proof of b.} Let $\xi \in Z_{i + 1}$. Call a pair $(x_0, m)$, where $x_0 \in \xi$ and $m \in \{1, 2, \ldots, \lceil 100 / \varepsilon \rceil \}$,
{\it valid} if the map $f(\xi')$ defined by~\eqref{eq:sheet} is a sheet of $\mathfrak f$ in some neighborhood of $\xi$. For  a.e. $\xi \in Z_{i + 1}$
the multiplicity $\# \mathfrak f(\xi)$ coincides with the number of valid pairs.

Assume a pair $(x_0, m)$ is valid. By construction of $\mathcal A$, there exists a point $y \in \xi$ such that
\begin{equation}\label{eq:0.9}
2 + 0.9 \varepsilon \leq \| y - x_0 \| \leq 2 + \varepsilon.
\end{equation}
But this is impossible for $x_0 \in \bigcup\limits_{i \in I_C} P_i$ because of the Forbidden Distances property of $\Delta$ and Lemma~\ref{lem:routine}, Assertion 1.
Therefore $x_0 \in P_i$, where $i \in I_A \cup I_B \cup I_D \cup I_E$. On the other hand, a straightforward volume estimate gives $(\xi \cap P_i) \leq \frac{(10K + 2)^d}{|B_1(\mathbf 0)|}$.
Applying Lemma~\ref{lem:card_i} finishes the proof of {\bf b}.

Now, as {\bf a} and {\bf b} are verified, let
$$c_0 = \left( \left( \frac{K}{K - 20} \right)^{d - 1} \cdot \left\lceil \frac{100}{\varepsilon} \right\rceil \cdot \frac{6 \Delta_0}{c} \cdot \frac{(10K + 2)^d}{|B_1(\mathbf 0)|}  \right)^{-1}.$$
Then~\eqref{eq:shrink} holds, because
$$\Pr (\eta \in Z_i \setminus Z^{term}_i) \leq \Pr(\eta \in \mathfrak f(Z_{i + 1})) \leq c_0^{-1} \Pr(\eta \in Z_{i + 1}),$$
where the second inequality is implied by {\bf a} and {\bf b}. \end{proof}

\section{Large Circuit Lemma}

\subsection{Statement of the lemma and reduction to two cases}

At this point and for this entire section we fix $\varepsilon > 0$ so that the function $\Delta$ defined by~\eqref{eq:defect_defn} satisfies the defect function
properties at level $\varepsilon$. We proceed by introducing the notion of a {\it large-circuit configuration} for a two-dimensional square box $Q_L$.

\begin{defn}\label{def:large_circuit}
We say that a configuration $\xi \in \Omega(\mathbb R^2)$ is a {\it large-circuit configuration} for in the square box
$Q_L$ (or {\it has a large circuit} in $Q_L$) if the graph $G_{\varepsilon}(\xi \cap (Q_L \setminus Q_{0.9L}))$ has a cycle $x_1 x_2 \ldots x_m x_1$ such that the polygonal
circuit $x_1 x_2 \ldots x_m x_1$ is contractible to $\partial Q_L$ in the square annulus $Q_L \setminus Q_{0.9L}$.
\end{defn}

The following Large Circuit Lemma is the main result of this section.

\begin{lem}[Large Circuit Lemma]
Let $p > 0$. Then there exist positive real numbers $q$ and $L$ such that uniform hard-disk model $\eta = \eta^{(s)}(Q_L, \zeta)$ in the square $Q_L \subset \mathbb R^2$
satisfies at least one of the following two assertions.
\begin{enumerate}
  \item[(LC1)] $\Pr \bigl( \exists y \in Q_L : B_{2 + \varepsilon}(y) \subset Q_L \; \text{and} \; B_{2 + \varepsilon}(y) \cap \eta = \varnothing \bigr) > q$.
  \item[(LC2)] $\Pr \bigl( \text{$\eta$ has a large circuit in $Q_L$} \bigr) > 1 - p$.
\end{enumerate}
\end{lem}

We prove the Large Circuit Lemma by reduction to two cases depending on the magnitude of $s$. Each of the two lemmas below corresponds to one of the cases.

\begin{lem}\label{lem:dense}
Let $p > 0$. Assume a constant $C > 0$ is given. Then there exist positive real numbers $q_0$ and $L_0$ such that the following holds.
If a two-dimensional uniform hard-disk model $\eta = \eta^{(s)}(Q_L, \zeta)$ satisfies
$$ L \geq L_0, \quad  s > \frac{(2L)^2}{2\sqrt{3}} - CL$$
then either (LC1) or (LC2) holds.
\end{lem}

\begin{lem}\label{lem:sparse}
There exist positive real numbers $C$, $L'_0$ and a positive-valued function $q'_0 : (L'_0, +\infty) \to \mathbb R$ such that the following holds.
If a two-dimensional uniform hard-disk model $\eta = \eta^{(s)}(Q_L, \zeta)$ satisfies
$$  L \geq L_0, \quad s \leq \frac{(2L)^2}{2\sqrt{3}} - CL,$$
then one has
$$\Pr \bigl( \exists y \in Q_L : B_{2 + \varepsilon}(y) \subset Q_L \; \text{and} \; B_{2 + \varepsilon}(y) \cap \xi = \varnothing \bigr) > q'_0(L).$$
\end{lem}

The Large Circuit Lemma indeed follows from Lemma~\ref{lem:dense} and Lemma~\ref{lem:sparse} as below.

\begin{proof}[Proof of the Large Circuit Lemma]
Let $C, L'_0$ be as provided by Lemma~\ref{lem:sparse}. Inserting this value of $C$ into Lemma~\ref{lem:dense} yields some $q_0$ and $L_0$. Let
$$L = \max(L_0, L'_0), \quad q = \min(q_0, q'_0(L)).$$
We show that the Large Circuit Lemma holds with this choice of $q$ and $L$.

Indeed, if $s \leq \frac{(2L)^2}{2\sqrt{3}} - CL$ then (LC1) holds by Lemma~\ref{lem:sparse}. If $s > \frac{(2L)^2}{2\sqrt{3}} - CL$ then either (LC1) or (LC2)
holds by Lemma~\ref{lem:dense}. \end{proof}

\subsection{Preliminaries}

Before we proceed with proofs of Lemmas~\ref{lem:dense} and~\ref{lem:sparse}, we provide two statements that are important for the further argument.

\begin{lem}\label{lem:finite_gibbs}
Let $\rho > 2$ and let $D_0 \subseteq \mathbb R^d$ be an open bounded domain. Assume that $m$ pairs
$(D_1, D'_1), (D_2, D'_2) \ldots, (D_m, D'_m)$ of open domains are given such that $D_i \subseteq D'_i \subseteq D_0$ and
$\dist(D'_i, D'_j) > \rho$ for every $1 \leq i < j \leq m$. Let, finally,
$E_1, E_2, \ldots, E_m \subset \Omega(\mathbb R^d)$ be measurable sets of configurations such that each indicator function $\mathbbm 1_{E_i}(\xi)$
depends only on $\xi \cap (D'_i)_{\rho}$. For every $\xi \in \Omega(\mathbb R^d)$ write
\begin{eqnarray*}
  \operatorname{P}_i (\xi) & = & \Pr(\eta^{(s_i)}(D_i, \xi) \in E_i), \quad \text{where $s_i = \# (\xi \cap D_i)$,} \\
  \operatorname{P}_i^{[\lambda]} (\xi) & = & \Pr(\eta^{[\lambda]}(D_i, \xi) \in E_i)
\end{eqnarray*}
Then:
\begin{enumerate}
  \item (Uniform case.) Every uniform hard-disk model $\eta^{(s)}(D_0, \zeta)$ satisfies
        $$\Pr \left( \eta^{(s)}(D_0, \zeta) \in \bigcap\limits_{i = 1}^{m} E_i \right) = \operatorname{E}\left[ \prod\limits_{i = 1}^{m} \operatorname{P}_i (\eta^{(s)}(D_0, \zeta))\right].$$
  \item (Poisson case.) Every Poisson hard-disk model $\eta^{[\lambda]}(D_0, \zeta)$ satisfies
        $$\Pr \left( \eta^{[\lambda]}(D_0, \zeta) \in \bigcap\limits_{i = 1}^{m} E_i \right) = \operatorname{E} \left[\prod\limits_{i = 1}^{m} \operatorname{P}_i^{[\lambda]} (\eta^{(s)}(D_0, \zeta))\right].$$
\end{enumerate}
\end{lem}

\begin{proof}
Denote $D = D_1 \cup D_2 \cup \ldots \cup D_m$ and $E = E_1 \cap E_2 \cap \ldots \cap E_m$.

Consider the uniform case. Write $\eta = \eta^{(s)}(D_0, \zeta)$
Define a map $T : \Omega(\mathbb R^d) \to \Omega(\mathbb R^d) \times \mathbb Z^n$ as follows:
$$T(\xi) = (\xi \setminus D', \#(\xi \cap D_1), \# (\xi \cap D_2), \ldots, \# (\xi \cap D_m).$$
For every $( \xi_0, s_1, s_2, \ldots, s_m) \in T(\Omega(\mathbb R^d))$  denote
$$\operatorname{P}(\xi_0, s_1, s_2, \ldots, s_m) = \left( \eta \in E \mid T(\eta) = ( \xi_0, s_1, s_2, \ldots, s_m) \right),$$
where the expression on the right-hand side is the regular conditional probability spanned by the map $T$.

Then the following holds:
\begin{eqnarray}
  \Pr \left( \eta \in E \right) & = & \operatorname{E} \left[ \operatorname{P}(T(\eta)) \right] \label{eq:markov_1} \\
  & = & \operatorname{E}\left[ \prod\limits_{i = 1}^{m} \operatorname{P}_i (\eta)\right]. \label{eq:markov_2}
\end{eqnarray}

Indeed,~\eqref{eq:markov_1} is the law of total expectation with $\Pr(\eta \in E)$ treated as $\operatorname{E} \left[ \mathbbm 1_E(\eta) \right]$. Further,
the regular conditional probability distribution of $\eta$, conditioned on $T(\eta) = T(\xi)$, coincides with the distribution of the random configuration
$$\theta(\xi) = (\xi \setminus D) \cup (\eta_1(\xi) \cap D_1) \cup (\eta_2(\xi) \cap D_2) \cup \ldots \cup (\eta_n(\xi) \cap D_m),$$
where all $\eta_i(\xi)$ are independent and each $\eta_i(\xi)$ is distributed as $\eta^{\left( \# (\xi \cap D_i) \right)}(D_i, \xi)$. But
$$\Pr(\theta(\xi) \in E) = \Pr(\eta_i(\xi) \in E_i \; \forall i) = \prod\limits_{i = 1}^{m} \operatorname{P}_i (\xi),$$
hence~\eqref{eq:markov_2} holds.

The proof of the Poisson case is essentially the same; thus the details are omitted. \end{proof}

\begin{rem*}
The lemma above is an instance of the so-called spatial Markov property. For instance, an analogous framework for the Ising model can be found in
[Friedli--Velenik, Exercise 3.11] and the remark following the cited exercise.
\end{rem*}

\begin{lem}\label{lem:index_elimination}
Let $I$ be a finite set of indices, $J \subseteq I$. Assume that each index $i \in I$ is supplied with a number $a_i \in [0, 1]$. Finally, let
$\mathcal F \subseteq 2^I$ be a family of subsets of $I$ such that the incidence $X \in \mathcal F$ implies that all supersets of $X$ are in $\mathcal F$, too.
Let $\mathcal F_J = \mathcal F \cap 2^J$. Then the following inequality holds.
$$\sum\limits_{X \in \mathcal F_J} \left( \prod\limits_{i \in X} a_i \prod\limits_{i \in J \setminus X} (1 - a_i) \right) \leq
\sum\limits_{X \in \mathcal F} \left( \prod\limits_{i \in X} a_i \prod\limits_{i \in I \setminus X} (1 - a_i) \right).$$
\end{lem}

\begin{proof}
Clearly, we have
$$
\sum\limits_{X \in \mathcal F} \left( \prod\limits_{i \in X} a_i \prod\limits_{i \in I \setminus X} (1 - a_i) \right) \geq
\sum\limits_{X \in \mathcal F_J}\sum\limits_{Y \in 2^{I \setminus J}} \left( \prod\limits_{i \in X \cup Y} a_i \prod\limits_{i \in I \setminus (X \cup Y)} (1 - a_i) \right).
$$

On the other hand, for every $X \in \mathcal F_J$ one has
\begin{multline*}
\sum\limits_{Y \in 2^{I \setminus J}} \left( \prod\limits_{i \in X \cup Y} a_i \prod\limits_{i \in I \setminus (X \cup Y)} (1 - a_i) \right) = \\
\prod\limits_{i \in X} a_i \prod\limits_{i \in I \setminus X} (1 - a_i) \cdot \sum\limits_{Y \in 2^{I \setminus J}} \left( \prod\limits_{i \in Y} a_i \prod\limits_{i \in (I \setminus J) \setminus Y} (1 - a_i) \right) = \\
\prod\limits_{i \in X} a_i \prod\limits_{i \in I \setminus X} (1 - a_i) \prod \limits_{i \in I \setminus J} (a_i + (1 - a_i)) = \prod\limits_{i \in X} a_i \prod\limits_{i \in I \setminus X} (1 - a_i).
\end{multline*}
Hence the lemma follows.
\end{proof}

\begin{rem*}
Lemma~\ref{lem:index_elimination} may be considered as a special, simpler, case of the Harris inequality~\cite{Har}. But this case allows for a short proof,
which we provide for the convenience of a reader.
\end{rem*}

\subsection{The case of large $s$}

This is the harder of the two cases considered. Before we proceed with the proof of Lemma~\ref{lem:dense}, let us recall the assumptions
we have already accepted. Namely, we assume that a given function $\Delta(\xi, \xi', D)$ satisfies the defect function properties at
a fixed level $\varepsilon \in (0, 1)$ (see Definition~\ref{def:defect_prop}). The variables $\rho, c, C$ have the same meaning as in Definition~\ref{def:defect_prop}.
In addition, we will use the parameter $K$ defined by the Thin Box Lemma.

The argument is arranged as follows: for every $L > 10^4 K$ we determine the range of pairs $(p, q)$ sufficient to ensure either (LC1) or (LC2) as in the statement
of the Large Circuit Lemma. We argue as if the boundary conditions $\zeta$ are fixed, however, none of the estimates below depends on $\zeta$.

In the above notation, let a positive integer $n$ satisfy
$$20nK < L \leq 20(n + 1)K.$$
(By assumption on $L$, we have $n \geq 500$.)
Consider the two families, $R_{i*}$ and $R'_{i*}$ of vertical rectangular boxes and the two families, $R_{*i}$ and $R'_{*i}$ of horizontal
rectangular boxes defined as follows:
\begin{equation}\label{eq:boxes}
\begin{array}{lll}
R'_{i*} & = & [(20i - 5)K, (20i + 15)K] \times [-20nK, 20nK], \\
R_{i*} & = & [(20i - 1)K, (20i + 1)K] \times [-20nK, 20nK], \\
R'_{*i} & = & [-20nK, 20nK] \times [(20i - 5)K, (20i + 5)K],\\
R_{*i} & = & [-20nK, 20nK] \times [(20i - 1)K, (20i + 1)K],
\end{array}
\end{equation}
where $i$ runs through the set $\{ -n + 1, -n + 2, \ldots, n - 1 \}$.
Clearly, $R_{i*} \subset R'_{i*}$. Further, a direct analogue of the Thin Box Lemma (obtained from the original Thin Box Lemma by an appropriate translation
or rotation) holds for every pair $(R_{i*}, R'_{i*})$ of vertical boxes as well as for every pair $(R_{*i}, R'_{*i})$ of horizontal boxes.

Given a configuration $\xi \in \Omega(\mathbb R^2)$, we will write $s_i(\xi) = \#(\xi \cap R_{i*})$.

Let us turn to the proof of Lemma~\ref{lem:dense}. We present the argument as a sequence of steps.

\noindent{\bf Step 1.} Construction of saturators.

This step is devoted to construction of saturated configurations that will be passed to the function $\Delta(\cdot, \cdot, \cdot)$ as the second
argument. Since our argument relies on the defect function properties, this step is crucial. Propositions~\ref{prop:saturator_1} and~\ref{prop:saturator_2}
below are the building blocks of our construction, while Proposition~\ref{prop:saturator_3} shows how these building blocks are combined to yield
a saturated configuration.

\begin{prop}\label{prop:saturator_1}
There exists a map
$$\theta_* : \Omega^(\mathbb R^2) \to \Omega(\mathbb R^2)$$
such that for every $\xi \in \Omega^(\mathbb R^2)$ the following properties hold:
\begin{enumerate}
  \item $\xi \cap \theta_*(\xi) = \varnothing$.
  \item $\xi \cup \theta_*(\xi) \in \Omega(\mathbb R^2)$.
  \item For each $i \in \{ -n + 1, -n + 2, \ldots, n - 1 \}$ one has $\dist (\theta_*(\xi), R_{i*}) \geq \rho$.
  \item If $y \in \mathbb R^2$, then one has either $\dist(y, \theta_*(\xi) \cup \xi) \leq 2$, or $\dist(y, R_{i*}) < \rho$ for some $i \in \{ -n + 1, -n + 2, \ldots, n - 1 \}$.
\end{enumerate}
\end{prop}

\begin{proof}
Assume that $\Hat{\theta}$ is a set satisfying the properties 1--3.
If the property 4 does not hold for $\Hat{\theta}$ and some point $y \in \mathbb R^2$, then one can replace $\Hat{\theta}$ by
$\Hat{\theta} \cup \{ y \}$, and the properties 1--3 will still hold. We call such replacement an {\it elementary expansion} of
$\Hat{\theta}$ by $y$.

Let us proceed by constructing $\theta_*$ using a greedy algorithm. Set $\Hat{\theta} = \varnothing$, as the properties 1--3
are evident. Next, perform as many consecutive expansions of $\Hat{\theta}$ by points $y \in B_1(\mathbf 0)$ as possible.
After that, perform as many consecutive elementary expansions of $\Hat{\theta}$ by points $y \in B_2(\mathbf 0)$ as possible.
Repeating this consecutively for each of the concentric balls $B_3(\mathbf 0), B_4(\mathbf 0), \ldots$ yields a limit shape
$\theta_*$ for $\Hat(\theta)$. One can see that $\theta_*$ is as required.
\end{proof}

\begin{rem*}
The question whether $\theta_*$ is measurable is not addressed in Proposition~\ref{prop:saturator_1}, because this is insufficient for the
further argument. However, the construction in the proof shows that $\theta_*$ can be constructed as a measurable map. The issue of
measurability, is, however, important for the next Proposition~\ref{prop:saturator_2}.
\end{rem*}

\begin{prop}\label{prop:saturator_2}
Let $\xi \in \Omega(\mathbb R^2)$, $i \in \{ -n + 1, -n + 2, \ldots, n - 1 \}$ and $s_i(\xi) = \#(\xi \cap R_{i*})$.
Then there exists a measurable map
$$\theta_i[\xi] : \Omega^{s_i}(R_{i*}, \xi) \to \Omega(\mathbb R^2)$$
such that for every $\eta \in \Omega^{s_i}(R_{i*}, \xi)$ the following holds:
\begin{enumerate}
  \item $\eta \cap \theta_i[\xi](\eta) = \varnothing$.
  \item $\eta \cup \theta_*(\xi) \cup \theta_i[\xi](\eta) \in \Omega(\mathbb R^2)$.
  \item For each $x \in \theta_i[\xi](\eta)$ one has $\dist (x, R_{i*}) < \rho$.
  \item If $y \in \mathbb R^2$, then one has either $\dist(y, \eta \cup \theta_*(\xi) \cup \theta_i[\xi](\eta)) \leq 2$,
	      or $\dist(y, R_{i*}) \geq \rho$.
\end{enumerate}
\end{prop}

\begin{proof}
The proof is similar to the one of Proposition~\ref{prop:saturator_1}. The measurability throughout the algorithm is maintained in a standard way.
\end{proof}

\begin{prop}\label{prop:saturator_3}
Let $\xi \in \Omega^{(s)}(Q_L, \zeta)$, $i \in \{ -n + 1, -n + 2, \ldots, n - 1 \}$, $s_i = \#(\xi \cap R_{i*})$. Define
$$\psi_i[\xi](\eta) = \eta \cup \theta_*(\xi) \cup \theta_i[\xi](\eta),$$
where $\eta$ runs through $\Omega^{s_i}(R_{i*}, \xi)$. Then $\psi_i[\xi]$ is an $(R'_{i*}, \rho)$-saturator over $\Omega^{(s)}(R_{i*}, \xi)$.
\end{prop}

\begin{proof}
Let $y \in \mathbb R^2$ satisfy $\dist(y, \psi_i[\xi](\eta)) > 2$. Then there exists an index $j \neq i$ such that $\dist (y, R_{j*}) \leq \rho$.
(This is an immediate consequence of the definitions of $\theta_*(\xi)$ and $\theta_i[\xi](\eta)$.) Therefore
$$\dist (y, R'_{i*}) \geq  10K - \rho > 2 \rho,$$
which finishes the proof.
\end{proof}

\noindent{\bf Step 2.} With $\xi$ fixed, the defects of $\psi_i[\xi]$ are uniformly bounded for most $i$.

The aim of this step is to prove certain consequences of the condition $s > \frac{(2L)^2}{2\sqrt{3}} - CL$ of Lemma~\ref{lem:dense}.

\begin{prop}\label{prop:counting}
Let $C > 0$. Then there exists $C' > 0$ such that the following holds. If the $(\xi, \xi', Q_L)$
is a defect-measuring triple in the domain of arguments of $\Delta$ and if the inequality $\# (\xi \cap Q_L) > \frac{|Q_L|}{2\sqrt{3}} - CL$
holds then one has
$$\Delta(\xi, \xi', Q_L) < C'L.$$
\end{prop}

\begin{proof}
Denote $s = \# (\xi \cap Q_L)$. Then, by the Counting property of $\Delta$,
$$\Delta(\xi, \xi', Q_L) \leq |Q_L| - s \cdot 2\sqrt{3} + C_{cnt}L \leq (2C\sqrt{3} + C_{cnt})L.$$
Hence it is sufficient to choose $C' = 2C\sqrt{3} + C_{cnt}$.
\end{proof}

Proposition~\ref{prop:counting} can be applied to the saturators $\psi_i[\xi]$ as follows.

\begin{prop}\label{prop:most_defects_bounded}
Let $C > 0$. Then there exists $C'' > 0$ such that the following holds. If a configuration $\xi \in \Omega(\mathbb R^2)$ satisfies
the inequality $\# (\xi \cap Q_L) > \frac{|Q_L|}{2\sqrt{3}} - CL$ and if $s_i = \#(\xi \cap R_{i*})$ then one has
$$ \# \{ i : \text{$i \in \{ -n + 1, -n + 2, \ldots, n - 1 \}$ and $(\xi, s_i) \in \Dfc_{C''}(R_{i*}, R'_{i*})$} \} \geq 2n - 1 - \frac{n}{20}.$$
\end{prop}

\begin{proof}
Let $I \subseteq \{ -n + 1, -n + 2, \ldots, n - 1 \}$ be the set of indices $i$ satisfying
$$(\xi, s_i) \notin \Dfc_{C''}(R_{i*}, R'_{i*}).$$
Assume $i \in I$. By definition of $\Dfc_{C''}(R_{i*}, R'_{i*})$, this means that the defect of $\eta^{(s_i)}(R_{i*}, \xi)$ with respect to $R_{i*}$ cannot be bounded by $C''$.
Thus, according to Definition~\ref{def:bd_defect} with $\phi = \psi_i[\xi]$, one can find a configuration $\eta_i \in \Omega^{(s_i)}(\xi, R_{i*})$
such that
\begin{equation}\label{eq:large_defect}
\Delta(\eta_i, \psi_i[\xi](\eta_i), R'_{i*}) > C''.
\end{equation}
Let~\eqref{eq:large_defect} be the definition of $\eta_i$ for $i \in I$. For $i \in \{ -n + 1, -n + 2, \ldots, n - 1 \} \setminus I$ let,
by definition, $\eta_i = \xi$.

Consider the configurations
$$\eta = \left( \xi \setminus \bigcup\limits_{i = -n + 1}^{n - 1} R_{i*} \right) \cup \bigcup\limits_{i = -n + 1}^{n - 1} (\eta_i \cup R_{i*}),$$
$$\eta' = \eta \cup \theta_*(\xi) \cup \bigcup\limits_{i = -n + 1}^{n - 1} \theta_i[\xi](\eta_i).$$
By construction, one concludes that $\eta \subseteq \eta'$ and that $\eta'$ is (globally) saturated.
Hence each $(\eta, \eta', R'_{i*})$ ($i = -n + 1, -n + 2, \ldots, n - 1$) is a defect-measuring triple and belongs to the domain of arguments of $\Delta$.

Further, for each $i \in I$ one has
$$\Delta(\eta, \eta', R'_{i*}) = \Delta(\eta_i, \psi_i[\xi](\eta_i), R'_{i*}) > C'',$$
as the first identity follows from the Localization property of $\Delta$.

Finally, we apply Proposition~\ref{prop:counting} to the configuration $\eta$. It is clear that
$$\# (\eta \cap Q_L) = \# (\xi \cap Q_L) > \frac{|Q_L|}{2\sqrt{3}} - CL.$$
Therefore
$$C' L \geq \Delta(\eta, \eta', R'_{i*}) \geq \sum\limits_{i \in I} \Delta(\eta, \eta', R'_{i*}) \geq C'' \cdot \# I.$$
Consequently,
$$\# I \leq \frac{C'}{C''}L \leq \frac{20(n + 1)KC'}{C''}.$$
Hence it is sufficient to choose $C'' = 10^4KC'$ to ensure $\# I \leq \frac{n}{20}$, as required.
\end{proof}

\begin{rem*}
Note that the Counting property of $\Delta$ was crucial for this step.
\end{rem*}

\noindent{\bf Step 3.} Small probability of an empty space event implies high probability of multiple crossings.

For the rest of this section we assume that the value of $C$ is inherited from the condition of
Lemma~\ref{lem:dense}. Thus, by Step 2, the values assigned to $C'$ and $C''$ become fixed as well.

This step essentially relies on the notion of the $(\varepsilon, \nu)$-crossing. Therefore the reader might find it helpful to
recall the notion from Definition~\ref{def:cross}.

In order to state and prove the main result of Step 3, Proposition~\ref{prop:alternative}, we will define a constant $p_1$ and a function $p_2 : \mathbb N \to \mathbb R$.
The following Proposition~\ref{prop:p1_def} serves as the definition of $p_1$.

\begin{prop}\label{prop:p1_def}
Let $\xi \in \Omega^{(s)}(Q_L, \zeta)$, where $s > \frac{|Q_L|}{2\sqrt{3}} - CL$.
Assume that an index $i \in \{ -n + 1, -n + 2, \ldots, n - 1 \}$ satisfies
$$(\xi, s_i) \in \Dfc_{C''}(R_{i*}, R'_{i*}),$$
where $s_i = s_i(\xi)$. Define $\nu = \left\lceil \frac{6C''}{c(\varepsilon)} \right\rceil$, where $c(\varepsilon)$ is inherited from the definition of a defect function.
Then there exists a number $p_1 > 0$, independent of $L$, such that at least one of the following holds:
\begin{enumerate}
  \item $\Pr(\text{$R'_{i*}$ is $(\varepsilon, \nu)$-crossed by $\eta^{(s_i)}(R_{i*}, \xi)$}) > p_1$.
  \item $\Pr \left( \text{$R_{i*} \setminus \bigcup\limits_{x \in \eta^{(s_i)}(R_{i*}, \xi)} B_2(x)$ contains an $\varepsilon$-ball} \right) > p_1$.
\end{enumerate}
\end{prop}

\begin{proof}
This is a direct consequence of the Thin Box Lemma.
\end{proof}

Let us define the function $p_2$. Let $n \in \mathbb N$ and $E_1, E_2, \ldots, E_n$ be independent events of probability $p_1$ each. We set, by definition,
$$p_2(n) = 1 - \Pr\left( \sum\limits_{i = 1}^{\left\lfloor \frac{n}{20} - 5 \right\rfloor} \mathbbm 1(E_i) > 8\nu \right).$$

\begin{prop}\label{prop:alternative}
Let $p, q > 0$ satisfy the inequality
\begin{equation}\label{eq:ineq_pq}
\frac{q}{p_1} + \frac{1 - \frac{p}{4}}{1 - p_2(n)} < 1.
\end{equation}
Then for every boundary conditions $\zeta \in \Omega(\mathbb R^2)$ and every integer $s > \frac{|Q_L|}{2\sqrt{3}} - CL$ at least
one of the following holds:
\begin{enumerate}
  \item $\Pr \left( \sum\limits_{i = -n + 1}^{-n + \left\lfloor \frac{n}{10} - 4 \right\rfloor} \mathbbm 1(\text{$R'_{i*}$ is $(\varepsilon, \nu)$-crossed by $\eta^{(s)}(Q_L, \zeta)$} \}) > 8\nu \right) > 1 - \frac{p}{4}$.
  \item $\Pr \left( \text{$Q_L \setminus \bigcup\limits_{x \in \eta^{(s)}(Q_L, \zeta)} B_2(x)$ contains an $\varepsilon$-ball} \right) > q$.
\end{enumerate}
\end{prop}

\begin{proof}
For each $i \in \{ -n + 1, -n + 2, \ldots, n - 1 \}$ define
$$F^{cross}_i = \{ \eta \in \Omega(\mathbb R^2) : \text{$R'_{i*}$ is $(\varepsilon, \nu)$-crossed by $\eta$} \},$$
$$F^{empty}_i = \left\{ \eta \in \Omega(\mathbb R^2) : \text{$R_{i*} \setminus \bigcup\limits_{x \in \eta} B_2(x)$ contains an $\varepsilon$-ball} \right\}.$$

Denote
\begin{multline*}
G = \{\xi \in \Omega^{(s)}(Q_L, \zeta) : \\
\text{$\Pr( \eta^{s_i(\xi)}(R_{i*}, \xi) \in F^{empty}_i ) > p_1$ holds for some index $i$} \}.
\end{multline*}

Consider the two cases.

\noindent {\bf Case 1.} $\Pr(\eta^{(s)}(Q_L, \zeta) \in G) > \frac{q}{p_1}$.

Let us apply Lemma~\ref{lem:finite_gibbs} with $D_0 = Q_L$, $m = 2n - 1$, $D_i = R_{(-n + i)*}$, $D'_i = R'_{(-n + i)*}$ and $E_i = \Omega(\mathbb R^2) \setminus F^{empty}_{-n + i}$.
For every $\xi \in G$ we have
$$\prod\limits_{i = 1}^{m} \operatorname{P}_i(\xi) \leq 1 - p_1,$$
because at least one multiplier does not exceed $1 - p_1$ and the others do not exceed 1. Consequently,
\begin{multline*}
\Pr \left( \text{$Q_L \setminus \bigcup\limits_{x \in \eta^{(s)}(Q_L, \zeta)} B_2(x)$ contains an $\varepsilon$-ball} \right) \geq \\
1 - \Pr\left( \eta^{(s)}(Q_L, \zeta) \in \bigcap\limits_{i = 1}^{m} E_i \right) = \\
1 - \operatorname{E} \left[ \prod\limits_{i = 1}^{n - 1} \operatorname{P}_i (\eta^{(s)}(Q_L, \zeta)) \right] \geq \\
1 - \Pr(\eta^{(s)}(Q_L, \zeta) \in G) \cdot (1 - p_1) - (1 - \Pr(\eta^{(s)}(Q_L, \zeta) \in G)) = \\
p_1 \cdot \Pr(\eta^{(s)}(Q_L, \zeta) \in G) > q,
\end{multline*}
which is exactly the second assertion of the alternative.

\noindent {\bf Case 2.} $\Pr(\eta^{(s)}(Q_L, \zeta) \notin G) > \frac{1 - \frac{p}{4}}{1 - p_2(n)}.$

Denote
$$I = \left\{ 1, 2, \ldots, \left\lfloor \frac{n}{10} - 4 \right\rfloor \right\}, \quad \mathcal F = \{ X \subseteq I : \# X > 8 \nu \},$$
$$s_i(\xi) = \# (\xi \cap R_{(-n + i)*}), \quad a_i(\xi) = \Pr(\eta^{(s_i)}(R_{(-n + i)*}, \xi) \in F^{cross}_{-n + i}), \quad \text{where $s_i = s_i(\xi)$.}$$

Consider an arbitrary set of integers $X \in \mathcal F$ and let us apply Lemma~\ref{lem:finite_gibbs} with $D_0 = Q_L$, $m = \left\lfloor \frac{n}{10} - 4 \right\rfloor$,
$D_i = R_{(-n + i)*}$, $D'_i = R'_{(-n + i)*}$ and
$$E_i = \left\{
\begin{array}{ll}
F^{cross}_{-n + i} & \text{if $i \in X$} \\
\Omega(\mathbb R^2) \setminus F^{cross}_{-n + i} & \text{if $i \notin X$.}
\end{array}
\right.$$
If
$$X(\xi) = \{i \in I : \text{$R'_{(-n + i)*}$ is $(\varepsilon, \nu)$-crossed by $\xi$}$$
and $\eta = \eta^{(s)}(Q_L, \zeta)$, one has
$$\Pr \left( X(\eta) = X \right) = \operatorname{E} \left[ \prod\limits_{i \in X} a_i(\eta) \prod\limits_{i \in I \setminus X} (1 - a_i(\eta)) \right].$$
Therefore
\begin{multline*}
\Pr \left( \sum\limits_{i = -n + 1}^{-n + \left\lfloor \frac{n}{10} - 4 \right\rfloor} \mathbbm 1(\text{$R'_{i*}$ is $(\varepsilon, \nu)$-crossed by $\eta$} \}) > 8\nu \right) = \\
\operatorname{E} \sum\limits_{X \in \mathcal F} \left( \prod\limits_{i \in X} a_i(\eta) \prod\limits_{i \in I \setminus X} (1 - a_i(\eta)) \right).
\end{multline*}

Further, denote
$$Y(\xi) = \{ i \in I : (\xi, s_i(\xi)) \in \Dfc_{C''}(R_{(-n + i)*}, R'_{(-n + i)*}) \}.$$
By Proposition~\ref{prop:most_defects_bounded}, every configuration $\xi \in \Omega^{(s)}(Q_L, \zeta)$ satisfies
$$\# J \geq \left\lfloor \frac{n}{10} - 4 \right\rfloor - \frac{n}{20} \geq \left\lfloor \frac{n}{20} - 5 \right\rfloor.$$

Let, in addition $\xi \notin G$. Then, by Proposition~\ref{prop:p1_def}, the inequality $a_i(\xi) > p_1$ is satisfied for every $i \in Y(\xi)$.
Therefore
\begin{multline*}
\sum\limits_{X \in \mathcal F} \left( \prod\limits_{i \in X} a_i(\xi) \prod\limits_{i \in I \setminus X} (1 - a_i(\xi)) \right) \geq \\
\sum\limits_{X \in \mathcal F \cap 2^{Y(\xi)}} \left( \prod\limits_{i \in X} a_i(\xi) \prod\limits_{i \in Y(\xi) \setminus X} (1 - a_i(\xi)) \right)
\geq 1 - p_2(n),
\end{multline*}
where the first inequality follows from Lemma~\ref{lem:index_elimination} and the second one follows from the definition of $p_2$.

Hence
\begin{multline*}
\Pr \left( \sum\limits_{i = -n + 1}^{-n + \left\lfloor \frac{n}{10} - 4 \right\rfloor} \mathbbm 1(\text{$R'_{i*}$ is $(\varepsilon, \nu)$-crossed by $\eta$}) > 8\nu \right) \geq \\
(1 - p_2(n)) \cdot \Pr(\eta \notin G) > 1 - \frac{p}{4},
\end{multline*}
which is exactly the first assertion of the alternative.

But
$$1 = \Pr(\eta^{(s)}(Q_L, \zeta) \in G) + \Pr(\eta^{(s)}(Q_L, \zeta) \notin G) = 1 > \frac{q}{p_1} + \frac{1 - \frac{p}{4}}{1 - p_2(n)}.$$
Therefore Case 1 and Case 2 exhaust all possibilities.
\end{proof}

\noindent{\bf Step 4.} Multiple vertical and horizontal crossings guarantee a large circuit.

At this point the reader might find it helpful to recall the notation $R_{*i}$, $R'_{*i}$ (for vertical thin boxes) and $R_{i*}$, $R'_{i*}$ (for horizontal thin boxes),
introduced in~\eqref{eq:boxes}. We use the notation in the following Proposition~\ref{prop:circuit}, which is the main result of this step,

\begin{prop}\label{prop:circuit}
Let $\xi \in \Omega(\mathbb R^2)$. Assume that there are four sets of indices
\begin{multline*}
  I_{left}, I_{lower} \subseteq \left\{ -n + 1, -n + 2, \ldots, -n + \left\lfloor \frac{n}{10} - 4 \right\rfloor \right\}, \\
  I_{right}, I_{upper} \subseteq \left\{ n - 1, n - 2, \ldots, n - \left\lfloor \frac{n}{10} - 4 \right\rfloor \right\},
\end{multline*}
such that
$$\min (\# I_{left}, \# I_{right}, \# I_{lower}, \# I_{upper}) > 8\nu,$$
all the boxes $R'_{i*}$ are $(\varepsilon, \nu)$-crossed by $\xi$ for $i \in I_{left} \cup I_{right}$, and all the boxes $R'_{*i}$ are $(\varepsilon, \nu)$-crossed by $\xi$ for
$i \in I_{lower} \cup I_{upper}$. Then $\xi$ has a large circuit in $Q_L$.
\end{prop}

\begin{proof}
With no loss of generality, assume that
$$\# I_{left} = \# I_{right} = \# I_{lower} = \# I_{upper} = 8\nu + 1.$$
Define $P_{ij} = R'_{i*} \cap R_{*j}$. If $i \in I_{left} \cup I_{right}$, $j \in I_{lower} \cup I_{upper}$ and the square $P_{ij}$ is exceptional
either for $R'_{i*}$ or for $R'_{*j}$, then the pair $(i, j)$ will be called exceptional, too. It is clear that the number of exceptional pairs, $m_{exc}$
satisfies the inequality
$$m_{exc} \leq 4\nu(8\nu + 1).$$

Now choose $i_1 \in I_{left}$, $i_2 \in I_{right}$, $j_1 \in I_{lower}$ and $j_2 \in I_{upper}$ at random and independently. One can see that the expected number
of exceptional pairs $(i_k, j_l)$, where $(k, l)$ runs through $\{1, 2\}^2$, equals $\frac{m_{exc}}{(8\nu + 1)^2} < 1$.
Hence there exist $i_1, i_2, j_1, j_2$ as above such that none of the pairs $(i_k, j_l)$ is exceptional. Consequently, there exists a large circuit for $Q_L$ enclosed in
the set $R'_{i_1*} \cup R'_{*j_1} \cup R'_{i_2*} \cup R'_{*j_2}$ (by definitions of crossing and exceptionality).
\end{proof}

\noindent{\bf Step 5.} Assembling the proof of Lemma~\ref{lem:dense}.

\begin{proof}[Proof of Lemma~\ref{lem:dense}]
It is clear that $\lim\limits_{n \to \infty} p_2(n) = 0$. Hence, given $p > 0$, one can choose $q > 0$ and $n_0 \in \mathbb N$
such that the inequality~\eqref{eq:ineq_pq} holds for every $n > n_0$. We will show that the statement of Lemma~\ref{lem:dense} holds with $q_0 = q$ and $L_0 = 20n_0K$.

Let a hard-disk model $\eta = \eta^{(s)}(Q_L, \zeta)$ satisfy the conditions of Lemma~\ref{lem:dense}. If assertion (LC1) of the Large Circuit Lemma holds, there is
nothing to prove. Therefore we assume that (LC1) is false.

For an arbitrary configuration $\xi \in \Omega(\mathbb R^2)$ define
\begin{multline*}I_{left}(\xi) = \biggl\{ i \in \left\{ -n + 1, -n + 2, \ldots, -n + \left\lfloor \frac{n}{10} - 4 \right\rfloor \right\} : \\
\text{$R'_{i*}$ is $(\varepsilon, \nu)$-crossed by $\xi$} \biggr\}.
\end{multline*}
By Proposition~\ref{prop:alternative}, we have
\begin{equation}\label{eq:left_crossings}
\Pr(\# I_{left}(\eta^{(s)}(Q_L, \zeta)) > 8\nu) > 1 - \frac{p}{4}.
\end{equation}
Let us define $I_{right}(\xi), I_{lower}(\xi)$ and $I_{upper}(\xi)$ in a similar way to $I_{left}(\xi)$. Then the inequalities similar to~\eqref{eq:left_crossings} apply.
The union bound for these inequalities yields
$$\Pr(\min (\# I_{left}(\eta), \# I_{right}(\eta), \# I_{lower}(\eta), \# I_{upper}(\eta) > 8\nu) > 1 - p.$$

By Proposition~\ref{prop:circuit}, assertion (LC2) of the Large Circuit Lemma follows.
\end{proof}

\subsection{The case of small $s$}

The goal of this subsection is to prove Lemma~\ref{lem:sparse}. We start with a simple observation.

\begin{prop}\label{prop:pckdensity_finite}
Let $L > 0$. Then there exists a configuration $\xi \in \Omega(\mathbb R^2)$ such that
$$\# (\xi \cap Q_L) \geq \frac{|Q_L|}{2\sqrt{3}}.$$
\end{prop}

\begin{proof}
Assume the converse. Then, since the $\# (\xi \cap Q_L)$ is an integer, one necessarily has
$$\# (\xi \cap Q_L) \leq |Q_L|\left( \frac{1}{2\sqrt{3}} - \delta \right)$$
for every $\xi \in \Omega(\mathbb R^2)$ and some absolute constant $\delta > 0$. Therefore the inequality
$$\# (\xi \cap (Q_L + t)) \leq |Q_L|\left( \frac{1}{2\sqrt{3}} - \delta \right)$$
holds for a fixed configuration $\xi$ and an arbitrary translation $t$. This immediately implies
$$\alpha(2) = \limsup\limits_{M \to \infty} \sup\limits_{\xi \in \Omega(\mathbb R^2)} \frac{\# (\xi \cap Q_M)}{|Q_M|} \leq \frac{1}{2\sqrt{3}} - \delta.$$
But it is well-known (see~\cite{FT}) that $\alpha(2) = \frac{1}{2\sqrt{3}}$. A contradiction finishes the proof.
\end{proof}

We proceed to the proof of Lemma~\ref{lem:sparse}.

\begin{proof}[Proof of Lemma~\ref{lem:sparse}]
It is clear that there exist constants $C > 0$ and $L'_0 > 10$ such that the inequality
$$\frac{|Q_{L - 7}|}{2\sqrt{3}} \geq \frac{|Q_L|}{2\sqrt{3}} - CL$$
holds for every $L > L'_0$.

Hence, by Proposition~\ref{prop:pckdensity_finite}, there is a set $\xi \in \Omega(\mathbb R^2)$ such that
$$\# (\xi \cap Q_{L - 7}) \geq \frac{|Q_L|}{2\sqrt{3}} - CL.$$
With no loss of generality, one can assume that $\xi \subset Q_{L - 7}$. Let
$$\xi = \{ x_1, x_2, \ldots, x_k \},$$
where $k \geq \frac{|Q_L|}{2\sqrt{3}} - CL$.

Let $0 \leq s \leq k$ be an integer, $\zeta \in \Omega(\mathbb R^2)$ be arbitrary boundary conditions.
Consider a subspace $E \subseteq \Omega^{(s)}(Q_L, \zeta)$ defined as follows:
$$E = \biggl\{(\zeta \setminus Q_L) \cup \{y_1, y_2, \ldots, y_s\} : \left \| y_i - \frac{L - 6}{L - 7}x_i \right\| < \frac{1}{2(L - 7)} \biggr\}.$$
Of course, $\inf\limits_{\zeta \in \Omega(\mathbb R^2)}\Pr( \eta^{(s)}(Q_L, \zeta) \in E ) > 0$.
In addition, there is a ball $B_3(z) \subset (Q_L \setminus Q_{L - 6})$. Since $\varepsilon < \varepsilon_0 \leq 1$, one concludes that
$$B_{2 + \varepsilon}(z) \subset Q_L \; \text{and} \; B_{2 + \varepsilon}(z) \cap \xi = \varnothing \quad \text{whenever $\xi \in E$.}$$
This finishes the proof of the lemma.
\end{proof}

\section{Proof of the Main Theorem}\label{sec:main_proof}

We proceed in three steps. First we transfer our results concerning the uniform model to the Poisson model.
Then we show that, with high probability, there exists a chain of points $t_1, t_2, \ldots, t_m \in L \cdot \mathbb Z^2$,
such that $\| t_{i + 1} - t_i \| = L$, $t_1 \in Q_M$, $t_m \notin Q_{M'}$ and each translate $\eta^{[\lambda]}(Q_{M' + 2L}, \zeta)$ has a large
circuit in $Q_L$. Finally, we prove that the annulus crossing indeed occurs whenever such a chain exists.

\noindent {\bf Step 1.} ``Poissonization'' of the Large Circuit Lemma.

The key statement of this step is Lemma~\ref{lem:pois_to_unif} below. In view of possible generalizations it is stated for arbitrary dimension $d$.

The reader might find it useful to recall the notion of a configuration admitting an empty $\varepsilon$-space in a bounded open domain $D$ (Definition~\ref{def:empty}).
Namely, $\xi$ admits an empty $\varepsilon$-space in $D$ if there exists an $\varepsilon$-ball $B_{\varepsilon}(w) \subseteq D$ such that $\dist(\xi, B_{\varepsilon}(w)) \geq 2$.

\begin{lem}\label{lem:pois_to_unif}
Let the dimension $d \geq 1$ be fixed and the numbers $p, q, \varepsilon, L > 0$ be given. Then there exists $\lambda_0 > 0$ such that the following holds.
If $E \subseteq \Omega(\mathbb R^2)$ is measurable and a Poisson hard-disk model $\eta^{[\lambda]}(Q_L, \zeta)$ satisfies the conditions
\begin{equation}\label{eq:pois_to_unif:pois}
\begin{aligned}
& \lambda > \lambda_0, \\
& \Pr(\eta^{[\lambda]}(Q_L, \zeta) \in E) \leq 1 - p,
\end{aligned}
\end{equation}
then there exists a uniform hard-disk model $\eta^{(s)}(Q_L, \zeta)$ (with the same boundary conditions $\zeta$) such that the two inequalities below hold simultaneously:
\begin{align}
  & \Pr(\eta^{(s)}(Q_L, \zeta) \in E) \leq 1 - \frac{p}{2} \label{eq:pois_to_unif:p} \\
  & \Pr \left( \text{$\eta^{(s)}(Q_L, \zeta)$ admits an empty $\varepsilon$-space in $Q_L$} \right) \leq q. \label{eq:pois_to_unif:q}
\end{align}
\end{lem}

\begin{proof}
It is clear that there exists a non-negative integer $s_0 = s_0(\zeta)$ such that the uniform hard-disk model $\eta^{(s)}(Q_L, \zeta)$ is well-defined
for $s \in \{ 0, 1, \ldots, s_0 \}$ and undefined for $s > s_0$. If  $\xi \in \eta^{(s)}(Q_L, \zeta)$, then
$$\{ B_1(x) : x \in \xi \cap Q_L \}$$
is a packing of balls in $Q_{L + 1}$, and therefore
$$s = \# (\xi \cap Q_L) \leq \frac{|Q_{L + 1}|}{\beta},$$
where $\beta = |B_1(\mathbf 0)|$. Thus $s_0 \leq \frac{|Q_{L + 1}|}{\beta}$.

We will show that the choice
\begin{equation}\label{eq:lambda_0}
\lambda_0 = \frac{2 |Q_{L + 1}|}{\beta^2 \varepsilon^d pq}
\end{equation}
is sufficient.

Let
$$S_1(\zeta) = \{ s \in \{ 0, 1, \ldots, s_0(\zeta) \} : \text{\eqref{eq:pois_to_unif:p} is false} \},$$
$$S_2(\zeta) = \{ s \in \{ 0, 1, \ldots, s_0(\zeta) \} : \text{\eqref{eq:pois_to_unif:q} is false} \}.$$
Assume, for a contradiction, that a Poisson hard-disk model $\eta^{[\lambda]}(Q_L, \zeta)$ satisfies~\eqref{eq:pois_to_unif:pois} and
$$S_1(\zeta) \cup S_2(\zeta) = \{ 0, 1, \ldots, s_0(\zeta) \}.$$
With no loss of generality suppose $\zeta \cap Q_L = \varnothing$, since the replacement $\zeta \mapsto \zeta \setminus Q_L$ has no effect on any of the relevant models.
For the rest of the proof we will use the shortened notation $\eta$ for the Poisson model $\eta^{[\lambda]}(Q_L, \zeta)$.

A Poisson hard-disk model is known to be a weighted mixture of uniform hard-disk models. The weight assigned to a uniform model $\eta^{(s)}(Q_L, \zeta)$
($s \in \{ 0, 1, \ldots, s_0(\zeta) \}$) equals $\Pr(\# (\eta \cap Q_L)  = s)$ and satisfies the following expression:
$$\Pr(\# (\eta \cap Q_L)  = s) = \frac{\frac{\lambda^s}{s!} A_s(\zeta)}{\sum\limits_{i = 0}^{s_0(\zeta)}\frac{\lambda^i}{i!} A_i(\zeta)},$$
where
$$A_s(\zeta) = \int\limits_{(Q_L)^s} \mathbbm{1}(\{x_1, x_2, \ldots, x_s \} \cup (\zeta \setminus Q_L) \in \Omega(\mathbb R^d))\, dx_1 dx_2 \ldots dx_s.$$
(See, for instance,~\cite[Section 2]{Ar}.)

Consider an arbitrary integer $s \in S_2(\zeta)$. By definition of $S_2(\zeta)$, the space $\Omega^{(s)}(Q_L, \zeta)$
contains configurations that can be extended to a larger configuration by adding one point from $Q_L$. Therefore $s + 1 \leq s_0$. Moreover, if
$$F^{empty}(\zeta) = \{ \xi \in \Omega(\mathbb R^d) : \text{$\xi \setminus Q_L = \zeta$ and $\xi$ admits an empty $\varepsilon$-space in $Q_L$} \},$$
then
\begin{multline*}
A_{s + 1}(\zeta) = \int\limits_{(Q_L)^{s + 1}} \mathbbm{1}(\{x_1, x_2, \ldots, x_{s + 1} \} \cup \zeta \in \Omega(\mathbb R^d))\, dx_1 dx_2 \ldots dx_{s + 1} = \\
\int\limits_{(Q_L)^s} dx_1 dx_2 \ldots dx_s \Biggl[ \mathbbm{1}(\{x_1, x_2, \ldots, x_s \} \cup \zeta \in \Omega(\mathbb R^d)) \times \\
\int\limits_{Q_L} \mathbbm{1}(\dist(x_{s + 1}, \{x_1, x_2, \ldots, x_s \} \cup \zeta) \geq 2)\, dx_{s + 1} \Biggr] \geq \\
\int\limits_{(Q_L)^s} dx_1 dx_2 \ldots dx_s \Biggl[ \mathbbm{1}(\{x_1, x_2, \ldots, x_s \} \cup \zeta \in F^{empty}(\zeta)) \times \\
\int\limits_{Q_L} \mathbbm{1}(\dist(x_{s + 1}, \{x_1, x_2, \ldots, x_s \} \cup \zeta) \geq 2)\, dx_{s + 1} \Biggr] \geq \\
\beta \varepsilon^d \int\limits_{(Q_L)^s} \mathbbm{1}(\{x_1, x_2, \ldots, x_s \} \cup \zeta \in F^{empty}(\zeta)) \, dx_1 dx_2 \ldots dx_s \geq \beta \varepsilon^d q A_s.
\end{multline*}
Therefore
$$\Pr(\# (\eta \cap Q_L) = s + 1 ) \geq \frac{\lambda \beta \varepsilon^d q}{s + 1} \Pr(\# (\eta \cap Q_L) = s) \geq \frac{\lambda \beta \varepsilon^d q}{s_0} \Pr(\# (\eta \cap Q_L) = s).$$
Thus~\eqref{eq:lambda_0} and the assumption $\lambda > \lambda_0$ imply
$$\Pr(\# (\eta \cap Q_L) \in S_2(\zeta) ) \leq \frac{s_0}{\lambda \beta \varepsilon^d q} \leq \frac{1}{\lambda} \cdot \frac{|Q_{L + 1}|}{\beta^2 \varepsilon^d q} \leq \frac{p}{2}.$$

From the assumption $S_1(\zeta) \cup S_2(\zeta) = \{ s \in \{ 0, 1, \ldots, s_0(\zeta) \}$ we conclude
$$\Pr(\# (\eta \cap Q_L) \in S_1(\zeta) ) \geq 1 - \frac{p}{2}.$$
Consequently,
\begin{equation*}
\begin{aligned}
& \Pr(\eta \in E) \geq \\
& \sum\limits_{s \in S_1(\zeta)} \Pr \left( \eta \in E \bigm\vert \# (\eta \cap Q_L) = s \right) \cdot \Pr(\# (\eta \cap Q_L) = s) > \\
& \sum\limits_{s \in S_1(\zeta)} \left( 1 - \frac{p}{2} \right) \cdot \Pr(\# (\eta \cap Q_L) = s) = \\
& \left( 1 - \frac{p}{2} \right) \cdot \Pr(\# (\eta \cap Q_L) \in S_1) \geq \\
& \left( 1 - \frac{p}{2} \right)^2 > 1 - p.
\end{aligned}
\end{equation*}
A contradiction to the second inequality of~\eqref{eq:pois_to_unif:pois} finishes the proof.\end{proof}

We apply Lemma~\ref{lem:pois_to_unif} obtaining the following Proposition~\ref{prop:large_circuit_pois} on the probability to see a large circuit.
For the explicit definition of a large circuit the reader may refer to Definition~\ref{def:large_circuit}.

\begin{prop}[Large Circuit Lemma, Poisson version]\label{prop:large_circuit_pois}
Let $\varepsilon, p, \rho > 0$. Then there exists $L > 2\rho$ and $\lambda_0 > 0$ such that the inequality
$$\Pr \left( \text{$\eta^{[\lambda]}(Q_L, \zeta)$ has a large circuit in $Q_L$} \right) > 1 - p$$
holds for all boundary conditions $\zeta \in \Omega(\mathbb R^2)$ and all $\lambda > \lambda_0$.
\end{prop}

\begin{proof}
Follows immediately from the Large Circuit Lemma and Lemma~\ref{lem:pois_to_unif}.
\end{proof}

\noindent {\bf Step 2.} Large circuits of fixed size percolate in the sense of a certain discrete Peierls-type lemma.

Again, the key statement of this step, Lemma~\ref{lem:peierls} is stated for an arbitrary dimension $d$, since it could be useful in possible generalizations.

We start with several definitions.

\begin{defn}
A point set $\{ u_1, u_2, \ldots, u_m \} \subset \mathbb Z^d$ is called {\it neighbor-free} if $\| u_i - u_j \|_{L_{\infty}} > 1$ for every $1 \leq i < j \leq m$.
\end{defn}

\begin{rem*}
The definition above uses the $L_{\infty}$ distance. By the $L_{\infty}$ norm of a vector we mean, as usual, the largest absolute value of its coordinates.
\end{rem*}

\begin{defn}
Let a dimension $d \geq 2$, a real number $p \in (0, 1)$ and a positive integer $M$ be given. Let $\tau \subseteq \mathbb Z^d$ be a random point set.
Assume that for every neighbor-free set $\{ u_1, u_2, \ldots, u_m \} \subset \mathbb Z^d \cap Q_N$ the inequality
$$\Pr( \{ u_1, u_2, \ldots, u_m \} \subseteq \mathbb Z^d \setminus \tau ) \geq p^m$$
holds. Then the random set $\tau$ will be called {\it $p$-dense} in $Q_N$.
\end{defn}

\begin{defn}\label{def:discrete_crossing}
A point set $\tau \subseteq \mathbb Z^d$ will be called $(M, N)$-crossing, where $M, N \in \mathbb Z$ and $0 < M \leq N$, if there exists
a sequence of points $t_1, t_2, \ldots, t_k \in \tau$ such that $\| t_{i + 1} - t_i \| = 1$, $t_1 \in \partial Q_M$, $t_k \in \partial Q_N$.
\end{defn}

\begin{lem}\label{lem:peierls}
Let $d \geq 2$ be a fixed dimension. Then there exist positive real numbers $c_1, C_1, C_2 > 0$ such that the following holds.
If $N$ is a positive integer, $p < \min (1, 1 / C_2)$ is a positive real number and $\tau \subseteq \mathbb Z^d$ is a random point set
satisfying the $p$-density property in $Q_N$, then the inequality
$$\Pr( \text{$\tau$ is $(M, N)$-crossing} ) \geq 1 - C_1(C_2p)^{c_1M^{d - 1}}$$
holds for every $M \in \{1, 2, \ldots, N \}$.
\end{lem}

\begin{proof}
Denote
$$E_{M, N} = \{ \upsilon \subseteq \mathbb Z^d : \text{$\upsilon$ is not $(M, N)$-crossing} \}.$$
Let us construct a map $\sigma_{M, N} : E_{M, N} \to 2^{\mathbb Z^d}$ as follows:
\begin{multline*}
\sigma_{M, N}(\upsilon) = (\mathbb Z^d \setminus Q_{M'}) \cup \sigma'_{M, N}(\upsilon), \quad \text{where} \\
\sigma'_{M, N}(\upsilon) = \{ t \in \mathbb Z^d \cap Q_{M'} : \text{$\exists t_1, t_2, \ldots, t_k \in \upsilon$ such that} \\
\text{$t_1 = t$, $t_k \in \partial Q_{M'}$ and $\| t_{i + 1} - t_i \| = 1$} \}.
\end{multline*}

Take an arbitrary set $\sigma \in \sigma_{M, N}(E_{M, N})$. Let us surround each point of $\sigma$ with a closed unit cube. We will
consider the boundary $\mathcal C(\sigma)$ of the union of all such cubes, i.e.,
$$\mathcal C(\sigma) = \partial \cl \left( \bigcup \limits_{x \in \sigma} (Q_{0.5} + x) \right).$$
In the framework of~\cite[Section~5.3]{Ru} the surface $\mathcal C(\sigma)$ can be decomposed into the so-called {\it Peierls contours}.
Each contour is a union of {\it plaquettes}, where a plaquette is a facet of some unit cube $Q_{0.5} + y$, $y \in \mathbb Z^d$.
A contour possesses the structure of adjacency of plaquettes. Each plaquette is declared adjacent to $2(d - 1)$ other plaquettes of the same contour,
one adjacency for each $(d - 2)$-face of a plaquette. Two plaquettes adjacent by a $(d - 2)$-face share that face; the converse is not guaranteed:
two plaquettes of the same contour sharing a $(d - 2)$-face may be adjacent by that face, but also may be non-adjacent.

A point $x \in \mathbb R^d$ is said to be {\it inside} a contour $\mathcal C' \subseteq \mathcal C(\sigma)$ if every sufficiently generic ray from the point $x$ intersects
an odd number of plaquettes of $\mathcal C'$. Then there is a unique contour $\mathcal C_0(\sigma) \subseteq \mathcal C(\sigma)$ such that the origin $\mathbf 0$
is inside $\mathcal C_0(\sigma)$.

Next, choose an arbitrary contour $\mathcal C_0 \in \mathcal C_0 \circ \sigma_{M, N}(E_{M, N})$.
Let $\# \mathcal C_0$ denote the {\it size} of the contour $\mathcal C_0$, i.e., the number of its plaquettes.
Since no plaquette $\mathcal C$ lies inside $Q_M$, the entire cube $Q_M$ lies in the interior of $\mathcal C_0$. Therefore the following inequality holds:
$$\# \mathcal C_0 \geq c_2 M^{d - 1},$$
where $c_2$ is a positive number, depending only on $d$.

We will say that an integer point $y \in \mathbb Z^d$ {\it approaches $\mathcal C_0$ from inside} if the cube $Q_{0.5} + y$ is inside $\mathcal C_0$ and its boundary,
$\partial(Q_{0.5} + y)$ has a common plaquette with $\mathcal C_0$. Let $\chi(\mathcal C_0)$ denote the set of all integer points approaching $\mathcal C_0$ from inside. 
Then there exists a positive number $c_3$, depending only on $d$, such that
$$\# \chi(\mathcal C_0) \geq c_3 \# \mathcal C_0(\sigma).$$
Moreover, it is possible to choose an neighbor-free set $\chi'(\mathcal C_0) \subseteq \chi(\mathcal C_0)$ such that
$$\# \chi'(\mathcal C_0) \geq c_4 \# \mathcal C_0(\sigma),$$
where $c_4$ is a positive number, depending only on $d$.

Let $\upsilon \in E_{M, N}$ and $\mathcal C_0 \circ \sigma_{M, N}(\upsilon) = \mathcal C_0$. Then, clearly, $\chi(\mathcal C_0) \subseteq \mathbb Z^d \setminus \upsilon$,
and thus $\chi'(\mathcal C_0) \subseteq \mathbb Z^d \setminus \upsilon$. Therefore, by the $p$-density of a random set $\tau$ one has
$$\Pr \left( \text{$\tau \in E_{M, N}$ and $\mathcal C_0 \circ \sigma_{M, N}(\tau) = \mathcal C_0$} \right) \leq \Pr(\chi'(\mathcal C_0) \subseteq \mathbb Z^d \setminus \tau) \leq p^{c_4 \# \mathcal C_0}.$$
Taking the sum over all $\mathcal C_0 \in \mathcal C_0 \circ \sigma_{M, N}(E_{M, N})$, one obtains
$$\Pr \left( \tau \in E_{M, N} \right) \leq \sum\limits_{i \geq c_2 M^{d - 1}} \left( \# \{\mathcal C_0 \in \mathcal C_0 \circ \sigma_{M, N}(E_{M, N}) : \# \mathcal C_0 = i \} \cdot p^{c_4 i} \right),$$
assuming that the right-hand side converges. But
$$\# \{\mathcal C_0 \in \mathcal C_0 \circ \sigma_{M, N}(E_{M, N}) : \# \mathcal C_0 = i \} \leq C_3^i$$
see~\cite{Ru, LM, BB}). Consequently,
$$\Pr \left( \tau \in E_{M, N} \right) \leq \sum\limits_{i \geq c_2 M^{d - 1}} C_3^i \cdot p^{c_4i} \leq C_1(C_2p)^{c_1M^{d - 1}},$$
which finishes the proof.
\end{proof}

Now we turn to the corollary for the Poisson hard-disk model.

\begin{prop}\label{prop:discrete_crossing}
Let $\varepsilon, p, \rho > 0$. There exist $L > 2\rho$ and $\lambda_0 > 0$ such that the following holds. For every 2-dimensional Poisson
hard-disk model $\eta = \eta^{[\lambda]}(Q_{L'}, \zeta)$, where $L' > 10L$ and $\lambda > \lambda_0$, and every integer $M < \frac{L'}{L} - 2$ the random integer-point set
\begin{equation}\label{eq:tau}
\tau_{\varepsilon, L}(\eta) = \{ t \in \mathbb Z^2 : \text{$\xi - Lt$ has a large $\varepsilon$-circuit in $Q_L$} \}
\end{equation}
satisfies the property
\begin{equation}\label{eq:discrete_crossing}
\Pr\left( \text{$\tau_{\varepsilon, L}(\eta)$ is $\left(M, \left\lfloor L'/L \right\rfloor - 1 \right)$-crossing} \right) \geq 1 - C_1(C_2p)^{c_1M},
\end{equation}
where $c_1, c_2, C_1$ are positive absolute constants.
\end{prop}

\begin{proof}
Choose $L$ and $\lambda_0$ as in Proposition~\ref{prop:large_circuit_pois}.
By Lemma~\ref{lem:peierls}, it is sufficient to prove the $p$-density of the set $\tau_{\varepsilon, L}(\eta)$. But this is an immediate corollary
of Lemma~\ref{lem:finite_gibbs} for $D_i = Q_L + L u_i$ if $\{ u_1, u_2, \ldots, u_m \} \subseteq \mathbb Z^d$ is the relevant neighbor-free set.\end{proof}

\noindent {\bf Step 3.} The event on the left-hand side of~\eqref{eq:discrete_crossing} implies the annulus crossing.

\begin{prop}\label{prop:dc_to_ac}
Let a real number $L > 0$ two positive integers $M < N$ be given. Assume that a configuration $\xi \in \Omega(\mathbb R^2)$ is such that the
set $\tau_{\varepsilon, L}(\xi)$, defined by~\eqref{eq:tau}, is $(M, N)$-crossing. Then
$$\xi \in \AnnCross(\varepsilon, ML, NL),$$
i.e., there is a connected component of the graph $G_{\varepsilon}(\xi)$ with a vertex in $Q_{ML}$ and another vertex in $\mathbb R^2 \setminus Q_{NL}$.
\end{prop}

\begin{proof}
Let $t_1, t_2, \ldots, t_k \in \tau_{\varepsilon, L}(\xi)$ be the sequence of integer points as in Definition~\ref{def:discrete_crossing}.

The definition of a large circuit can be naturally extended from the one for the box $Q_L$ to every translate $Q_L + t$ of this box. Namely,
$\xi$ has a large circuit in $Q_L + t$ if $\xi - t$ has a large circuit in $Q_L$. If $x_1 - t, x_2 - t, \ldots, x_m - t, x_1 - t$ ($x_i \in \xi$) is
a large circuit for $\xi - t$ in the box $Q_L$ then we call $x_1, x_2, \ldots, x_m, x_1$ a large circuit for $\xi$ in the box $Q_L + t$.

Correspondingly, by the definition~\eqref{eq:tau} of $\tau_{\varepsilon, L}(\xi)$, for every $i = 1, 2, \ldots, k$ the configuration $\xi$ has
a large circuit $\mathfrak c_i \subset G_{\varepsilon}(\xi)$ in each square box $Q_L + t_i$.

Since $\|t_i - t_{i + 1} \| = 1$, the circuits $\mathfrak c_i$ and $\mathfrak c_{i + 1}$ intersect in the following sense: one can choose an edge $(x, y)$ of $\mathfrak c_i$
and an edge $(z, w)$ of $\mathfrak c_{i + 1}$ such that the segments $[x, y]$ and $[z, w]$ have a common point. Thus
$$\| x - z \| + \| y - w \| \leq \| x - y \| + \| z - w \| \leq 2(2 + \varepsilon).$$
Therefore $\min ( \| x - z\| , \| y - w \| ) \leq 2 + \varepsilon$, and, consequently, $\mathfrak c_i$ and $\mathfrak c_{i + 1}$ belong to the same connected component of
$G_{\varepsilon}(\xi)$.

As a conclusion, $\mathfrak c_1$ and $\mathfrak c_m$ are in the same connected component of $G_{\varepsilon}(\xi)$. But $\mathfrak c_1$ has a vertex in $Q_{ML}$,
and $\mathfrak c_m$ has a vertex in $\mathbb R^2 \setminus Q_{NL}$. Hence the lemma follows.
\end{proof}

We are ready to finish the proof of the Main Theorem.

\begin{proof}[Proof of the Main Theorem]
Choose $L$ and $\lambda_0$ as in Proposition~\ref{prop:large_circuit_pois}. Then, by Propositions~\ref{prop:discrete_crossing}~and~\ref{prop:dc_to_ac}, one has
$$\Pr\left( \eta^{[\lambda]}(Q_{L'}, \zeta) \in \AnnCross(\varepsilon, ML, NL) \right) \geq 1 - C_1(C_2p)^{c_1M},$$
where $N = \lfloor L' / L \rfloor - 1$, $0 < M < N$. Since $NL > L' - 2L$, one concludes
$$\Pr\left( \eta^{[\lambda]}(Q_{L'}, \zeta) \in \AnnCross(\varepsilon, L_1, L' - 2L) \right) \geq 1 - C_1(C_2p)^{c_1\lfloor L_1 / L \rfloor}.$$
Hence the Main Theorem follows.
\end{proof}

\section{Corollaries for Gibbs distributions}\label{sec:gibbs}

In this section we prove some corollaries of the Main Theorem related to the notion of a Gibbs distribution. These results are direct counterparts of~\cite[Theorems 2, 3]{Ar}.

We start with a standard definition of a Gibbs distribution.

\begin{defn}
Let $\lambda > 0$. A random configuration $\eta \in \Omega(\mathbb R^2)$ is said to comply with a {\it Gibbs distribution} with intensity $\lambda$ if
every measurable function $f : \Omega(\mathbb R^2) \to [0, \infty)$ and every $L > 0$ satisfy
\begin{equation}\label{eq:def:1}
  \operatorname{E} \bigl[ f(\eta) \bigr] =  \operatorname{E} \biggl[ \operatorname{E} \bigl[ f(\eta^{[\lambda]}(Q_L, \eta)) \bigr] \biggr].
\end{equation}
\end{defn}

\begin{rem*}
The identity~\eqref{eq:def:1} is a close analogue to the one in Lemma~\ref{lem:finite_gibbs}, Poisson case. For this reason, a Gibbs distribution can be considered
as a generalization of a hard-disk model.
\end{rem*}

Gibbs measures are known to exist for every intensity $\lambda > 0$, while their uniqueness remains an open problem.

We are interested in a standard percolation question: if a random configuration $\eta$ is sampled from a Gibbs distribution with intensity $\lambda$,
is it true that with high probability the graph $G_{\varepsilon}(\eta)$ has an infinite connected component? The following proposition
provides a convenient notation for dealing with this question.

\begin{prop}
Let
\begin{equation}\label{eq:nested}
A(\varepsilon, M) = \bigcap\limits_{M' > M_0} \AnnCross(\varepsilon, M, M').
\end{equation}
Then for every $\xi \in A(\varepsilon, M)$ there exists a point $x \in \xi \cap Q_M$, which is a vertex of an infinite connected component of the graph $G_{\varepsilon}(\xi)$.
\end{prop}

\begin{proof}
We argue by contradiction. Assume $\xi \in A(\varepsilon, M)$ is a counterexample to the proposition. Then for every $x \in \xi \cap Q_M$ there exists $M'(x)$ such that
all vertices of the connected component of $x$ in $G_{\varepsilon}(\xi)$ belong to $Q_{M'(x)}$. The set $\xi \cap Q_M$ is finite, therefore the value
$$\Hat{M}' = \sup\limits_{x \in \xi \cap Q_M} M'(x)$$
is finite. But then $\xi \notin \AnnCross(\varepsilon, M, \Hat{M}')$, and hence $\xi \notin A(\varepsilon, M)$. A contradiction finishes the proof.
\end{proof}

Now we turn to the main results of this section, Theorems~\ref{thm:2} and~\ref{thm:3}.

\begin{thm}\label{thm:2}
Let $\varepsilon > 0$. Then there exist positive numbers $\lambda_0$, $c$ and $C$ such that the inequality
$$\Pr (\eta \in A(\varepsilon, M)) \geq 1 - C \exp(-cM)$$
holds for every $M > 0$, every $\lambda > \lambda_0$ and a random configuration $\eta$ sampled from any Gibbs distribution on $\Omega(\mathbb R^2)$ with intensity $\lambda$.
\end{thm}

\begin{proof}
Since the formula~\eqref{eq:nested} is a representation of $A(\varepsilon, M)$ as the intersection of nested families of configurations
(i.e., $\AnnCross(\varepsilon, M, M') \subseteq \AnnCross(\varepsilon, M, M'')$ whenever $M < M'' < M'$), it is sufficient to prove that
$$\Pr (\eta \in \AnnCross(\varepsilon, M, M')) \geq 1 - C \exp(-cM)$$
for every $M' > M$.

Consider the indicator function $f(\xi) = \mathbbm{1}_{\AnnCross(\varepsilon, M, M')}(\xi)$. Applying~\eqref{eq:def:1} to the function $f$ yields
\begin{multline*}
\Pr (\eta \in \AnnCross(\varepsilon, M, M')) = \operatorname{E} \bigl[ \mathbbm{1}_{\AnnCross(\varepsilon, M, M')}(\eta) \bigr] = \\
\operatorname{E} \biggl[ \operatorname{E} \bigl[ \mathbbm{1}_{\AnnCross(\varepsilon, M, M')}(\eta^{[\lambda]}(Q_{M' + L_1}, \eta)) \bigr] \biggr] = \\
\operatorname{E} \biggl[ \Pr \bigl[ \eta^{[\lambda]}(Q_{M' + L_1}, \eta) \in  \AnnCross(\varepsilon, M, M') \bigr] \geq 1 - C \exp(-cM),
\end{multline*}
where the last inequality immediately follows from the Main Theorem.

On the other hand, the inclusion $\AnnCross(\varepsilon, M, M') \subseteq \AnnCross(\varepsilon, M, M'')$ holds whenever $M < M'' < M'$.
Therefore~\eqref{eq:nested} implies
$$\Pr (\eta \in A(\varepsilon, M)) = \inf\limits_{M' > M} \Pr (\eta \in \AnnCross(\varepsilon, M, M')) \geq 1 - C \exp(-cM).$$
\end{proof}

\begin{thm}\label{thm:3}
Let $\varepsilon > 0$. Then there exists $\lambda_0 > 0$ such that the identity
$$\mathsf P(\text{$G_{\varepsilon}(\eta)$ has an infinite connected component}) = 1$$
holds for every $\lambda > \lambda_0$ and a random configuration $\eta$ sampled from any Gibbs distribution on $\Omega(\mathbb R^2)$ with intensity $\lambda$.
\end{thm}

\begin{proof}
Denote
$$A(\varepsilon) = \{ \xi \in \Omega(\mathbb R^2) : \text{$G_{\varepsilon}(\eta)$ has an infinite connected component}.$$
Then
$$A(\varepsilon) = \bigcup\limits_{M > 0} A(\varepsilon, M).$$
Since $A(\varepsilon, M) \subseteq A(\varepsilon, M')$ whenever $0 < M < M'$, we have
$$\Pr (\eta \in A(\varepsilon)) = \sup\limits_{M > 0} \Pr (\eta \in A(\varepsilon, M)) = 1,$$
where the last identity follows immediately from Theorem~\ref{thm:2}.
\end{proof}

\section*{Acknowledgements}
The author is thankful to A.~Sodin, R.~Peled and N.~Chandgotia for discussion that helped improving the earlier versions of the paper.

\end{document}